%% file: main.tex
\documentclass[acmsmall,nonacm]{acmart}

\usepackage[shortlabels]{enumitem}
\usepackage{subcaption}
\usepackage{wrapfig}

\DeclareMathOperator*{\E}{\mathbb{E}}
\let\Pr\relax
\DeclareMathOperator*{\Pr}{\mathbb{P}}

\DeclareMathOperator*{\Cov}{\mathrm{Cov}}
\newtheorem{claim}{Claim}

\newtheorem{assumption}{Assumption}

\newcommand{\sizeof}[1]{\left\lvert{#1}\right\rvert}

\newcommand{\floor}[1]{\left\lfloor{#1}\right\rfloor}

\newcommand{\edit}[1]{\textcolor{magenta}{#1}}

\AtBeginDocument{%
  }

\setcopyright{none}
\copyrightyear{2022}
\acmYear{2022}
\acmDOI{XXXXXXX.XXXXXXX}

\acmConference[SIGMETRICS '23]{Make sure to enter the correct
  conference title from your rights confirmation emai}{June 19--23,
  2023}{Orlando, FL}
\acmPrice{15.00}
\acmISBN{978-1-4503-XXXX-X/18/06}

\begin{document}

\title{Detecting TCP Packet Reordering in the Data Plane}


\author{Yufei Zheng}
\email{yufei@cs.princeton.edu}
\affiliation{%
  \institution{Princeton University}
  \streetaddress{Department of Computer Science}
  \city{Princeton}
  \state{New Jersey}
  \country{USA}
  \postcode{08540}
}

\author{Huacheng Yu}
\email{yuhch123@gmail.com}
\affiliation{%
  \institution{Princeton University}
  \streetaddress{Department of Computer Science}
  \city{Princeton}
  \state{New Jersey}
  \country{USA}
  \postcode{08540}
}

\author{Jennifer Rexford}
\email{jrex@cs.princeton.edu}
\affiliation{%
  \institution{Princeton University}
  \streetaddress{Department of Computer Science}
  \city{Princeton}
  \state{New Jersey}
  \country{USA}
  \postcode{08540}
}

\renewcommand{\shortauthors}{Zheng, Yu and Rexford}

\begin{abstract}
Network administrators want to  detect TCP-level packet reordering to diagnose performance problems and attacks.
However, reordering is expensive to measure, because each packet must be processed relative to the TCP sequence number of its predecessor in the same flow.
Due to the volume of traffic, detection  should take place in the data plane as the packets fly by.
However, restrictions on the memory size and the number of memory accesses per packet make it impossible to design an efficient algorithm for pinpointing flows with heavy packet reordering.
In practice, packet reordering is typically a property of a \emph{network path}, due to a congested or flaky link.
Flows traversing the same path 
are correlated in their out-of-orderness, and aggregating out-of-order statistics 
at the IP prefix level provides useful diagnostic information.
In this paper, we present efficient algorithms
for identifying \emph{IP prefixes} with heavy packet reordering under memory restrictions. 
First, we sample as many flows as possible, regardless of their sizes, but only for a short period at a time.
Next, we separately monitor the large flows over long periods, in addition to the flow sampling.
In both algorithms, we measure at the flow level, and aggregate statistics and allocate memory at the prefix level.
Our simulation experiments, using packet traces from campus and backbone networks, and our P4 prototype show that our algorithms correctly identify $80\%$ of the prefixes with heavy packet reordering using moderate memory resources.
\end{abstract}

\maketitle

\input{intro}
\input{problem}
\input{measurement}
\input{algo}
\input{eval}

\input{related}
\input{conclusion}


\bibliographystyle{ACM-Reference-Format}
\bibliography{reference}

\appendix
\input{appendix}




\end{document}

%% file: intro.tex
\section{Introduction}

Transmission Control Protocol (TCP) performance problems are often associated with packet reordering. Packet loss, commonly caused by congested links, triggers TCP senders to retransmit packets, leading these retransmitted packets to appear out of order.
Also, the network itself can cause packet reordering, due to malfunctioning equipment or traffic splitting over multiple links~\cite{paxson1997end}.  TCP overreacts to inadvertent reordering by retransmitting packets that were not actually lost and erroneously reducing the sending rate~\cite{blanton2002making,paxson1997end}.  In addition,  reordering of acknowledgment packets muddles TCP's self-clocking property and induces bursts of traffic~\cite{bennett1999packet}.
Perhaps more strikingly, reordering can be a form of denial-of-service (DoS) attack. 
In this scenario, an adversary persistently reorders existing packets, or injects malicious reordering into the network, to make the goodput low or even close to zero, despite delivering all of the packets~\cite{aad2008impact,herzberg2010stealth}.

To diagnose performance problems and neutralize attacks, it is  crucial to detect packet reordering quickly and efficiently, e.g., on the order of minutes.
Due to the sheer volume of traffic, the detection of packet reordering should take place in the data plane of network devices as the packets fly by.
This is because each packet must be processed in conjunction with its predecessor in the same flow, which renders simple packet sampling (e.g., widely used technologies like NetFlow~\cite{netflow} and sFlow\cite{sflow}) insufficient.
For example, sampling one in a thousand packets, let alone one in a million, would rarely ever capture consecutive packets of the same flow.

Fortunately, simple packet-reordering statistics can be collected directly as part of high-speed packet processing, given the
emergence of programmable data planes, including software platforms like eBPF~\cite{ebpf} and DPDK~\cite{dpdk}, smart network interface cards~\cite{pensando-nic,xilinx}, and ASIC-based switches~\cite{broadcom,tofino,penando-switch}.
With flexible parsing, we can extract the header fields we need to analyze the packets in a flow, including the TCP flow identifier (source and destination IP addresses and port numbers), the TCP sequence number, and the TCP segment length.
Using arrays or dictionaries, we can keep state across successive packets of the same flow.
In addition, simple arithmetic operations allow us to detect reordering and count the number of out-of-order packets in a flow.

However, processing packets efficiently for high link speeds imposes significant constraints on memory: 

\begin{itemize}
    \item \textbf{Memory size:} Modern data planes have a limited amount of memory, especially compared to the number of concurrent flows on high-speed links.
    
    \item \textbf{Memory accesses:} Since memory bandwidth has not kept pace with link bandwidth, modern data planes can only access memory a few times per packet.
\end{itemize}

\noindent
Plus, network devices perform other tasks---packet forwarding, access control, and so on---that demand a share of the already limited memory resources. Furthermore, since the data plane has limited bandwidth for communicating with the control-plane software, we cannot offload monitoring tasks to the control plane. As such, we need to design compact data structures that work within these constraints.

The limitations on memory size and accesses make it fundamentally difficult to pinpoint individual flows with a large proportion of out-of-order packets.
Yet, identifying every affected flow is not necessarily what is important for network administrators. Packet reordering is typically a property of a network path, due to congested or flaky links.
As such, it is useful to report reordering at a coarser level, such as to identify the IP prefixes associated with performance problems.
Since routing is determined at the IP prefix level, 
a network administrator could choose to route the traffic for an IP prefix through providers whose paths are not experiencing significant reordering.
However, this does \emph{not} obviate the need to maintain state for at least some flows, as TCP packet reordering is still a flow-level phenomenon.
Fortunately, we can identify prefixes with heavy packet reordering without needing to track \emph{all} of the flows, because
packets traversing the same path at the same time are often correlated in their out-of-orderness. 
In the presence of equal-cost multi-path routing, the bottlenecks may often be in a subpath that is shared, such as the access point~\cite{meng2022achieving}.
Even if the bottleneck only occurs on one of the paths, as long as we sample enough packets, there would still be correlation, albeit weaker.

In this paper, we present data structures that detect and report packet-reordering statistics to the control plane. 
\begin{itemize}
\item
We first sample as many flows as possible, regardless of their sizes, but only for a short period at a time.
Capitalizing on the correlation, we can capture the extent of reordering in prefixes by observing only snippets of their flows.
This flow-sampling approach performs especially well when  given a small amount of memory.
\item
When more memory is available, we can further improve the accuracy by monitoring heavy flows over longer periods of time in a separate data structure, and only sampling the rest of the flows. 
\end{itemize}

\noindent
The interplay between measuring at the flow level and acting at the prefix level lies at the heart of this problem.
To decide which set of flows to monitor, we need to incorporate prefix identity in managing the data structures, which gives rise to the idea of allocating memory at the prefix level.

In what follows, \S~\ref{sec:prob} formulates the reordering problem and shows the hardness of identifying out-of-order heavy flows. 
\S~\ref{sec:traffic} introduces real-world traffic workload characteristics that motivate our algorithm design.
We also verify the correlation among flows from the same prefix through measurement results.
We elaborate on the flow-sampling approach for finding heavily reordered prefixes in \S~\ref{sec:algo}, and discuss its optimizations for further accuracy gains.
In \S~\ref{sec:eval}, we demonstrate that our algorithms are extremely memory-efficient and hardware-friendly.
We discuss related work in \S\ref{sec:related} and then conclude 
 our paper in \S~\ref{sec:conclusion}.


%% file: problem.tex
\section{Problem Formulation: Identify Heavy Out-of-Order IP Prefixes} \label{sec:prob}
\begin{figure}[t]
\centering
\includegraphics[width=0.47\textwidth]{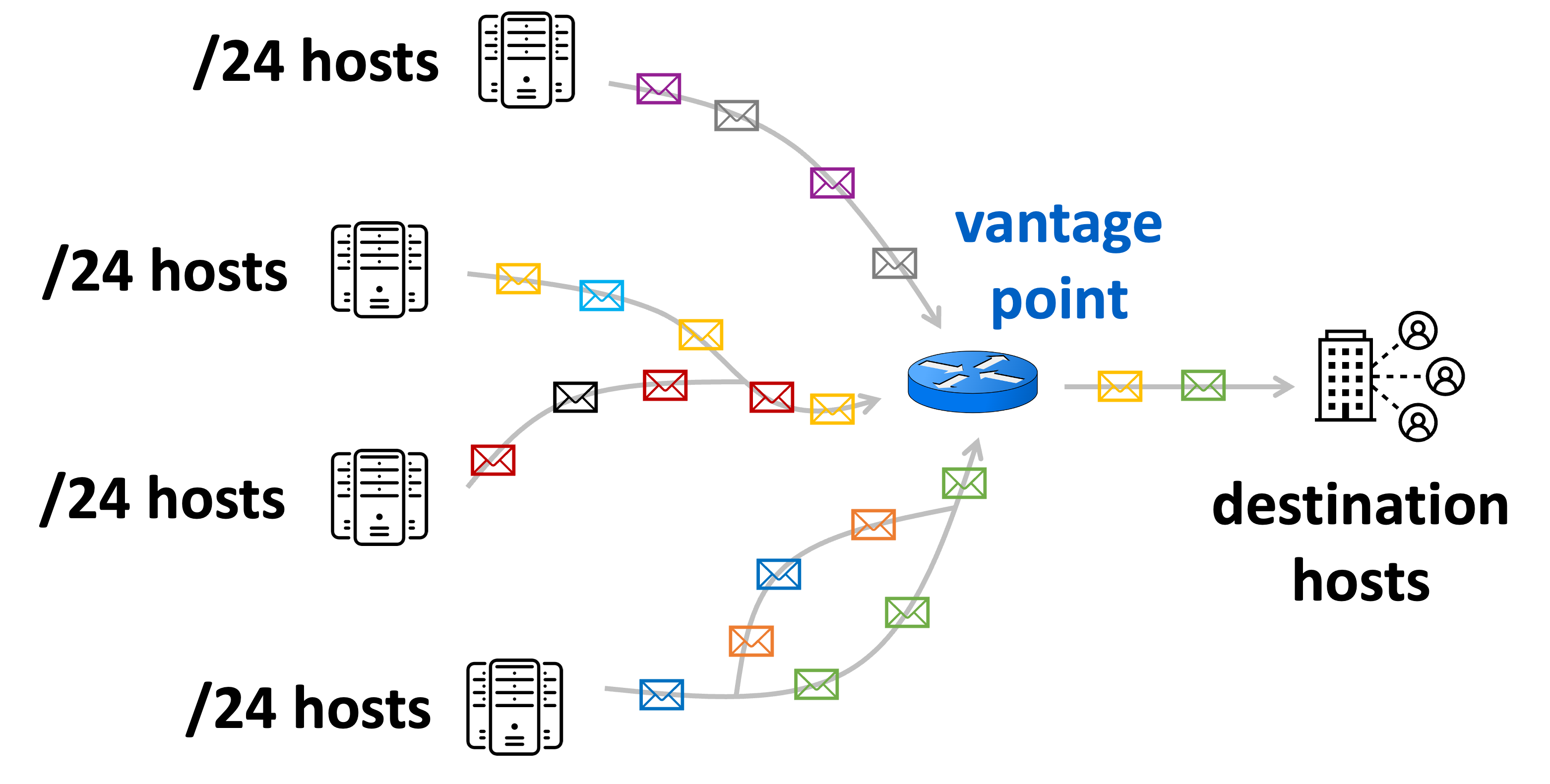}
\vspace{-0.8em}
\caption{
Different source prefixes send packets over different paths.
Packets on a path are colored differently to show that traffic from a single prefix has a mix of packets from different flows.
While flows from a single prefix may split over parallel subpaths, they do share many portions of their network resources.}
\label{fig:srcIPs}
\vspace{-1em}
\end{figure}

Consider a switch close to the receiving hosts, where we observe a stream of incoming packets (Figure~\ref{fig:srcIPs}).
Our goal is to identify the senders whose paths to the receivers are experiencing performance problems, through counting out-of-order packets.
In \S~\ref{subsec:flow}, we first introduce notations and definitions at the flow level, and show that identifying flows with heavy reordering is hard, even with randomness and approximation.
Later, in \S~\ref{subsec:prefix}, we extend the definitions to the prefix level, then discuss possible directions to identify heavy out-of-order prefixes.

\subsection{Flow-level reordering statistics}
\label{subsec:flow}

\subsubsection{Definitions at the flow level} \label{sec:prob:flow:def}

Consider a stream $S$ of TCP packets from different remote senders to the local receivers. 
In practice, TCP packets may contain payloads, and sequence numbers advance by the length of payload in bytes.
But, to keep the discussions simple, we assume sequence numbers advance by $1$ at a time, and we ignore sequence number rollovers.
We note that these assumptions can be easily adjusted to reflect the more realistic scenarios.
Then, a packet can be abstracted as a $3$-tuple $(f, s, t)$, with $f \in \mathcal{F}$ being its flow ID, $s \in [I]$ the sequence number and $t$ the timestamp.
In this case, a flow ID is a $4$-tuple of source and destination IP addresses, and the source and destination TCP port numbers.

Let $S_f = \{(f, s_i, t_i)\}_{i=1}^{N_f} \subseteq S$ be the set of packets corresponding to some flow $f$, sorted by time $t_i$ in ascending order.
We say the packets of flow $f$ are perfectly \emph{in-order} if $s_{i+1} = s_i + 1$ for all $i$ in $[N_f - 1]$.
By commonly used definitions, the $i$th packet in flow $f$ is \emph{out-of-order} if it has:
\begin{enumerate}[label={Def. \arabic*}, leftmargin = 3em] %
    \item \label{def:ooo_dec}
    a lower sequence number than its predecessor in $f$, $s_i < s_{i-1}$.
    \item \label{def:ooo_inc}
    a sequence number larger than that expected from its predecessor in $f$, $s_i > s_{i-1} + 1$.
    \item \label{def:ooo_max}
    a smaller sequence number than the maximum sequence number seen in $f$ so far, $s_i < \max_{j \in [i-1]}{s_j}$.
\end{enumerate}
When $s_i < s_{i-1}$ in flow $f$, we sometimes say an \emph{out-of-order event} occurs at packet $i$ with respect to ~\ref{def:ooo_dec}.
Out-of-order events with respect to other definitions are similarly defined.
Under each definition, denote the number of out-of-order packets in flow $f$ as $O_f$,
a flow $f$ is said to be \emph{out-of-order heavy} if $O_{f} > \varepsilon N_f$ for some small $\varepsilon > 0$.

In practice, none of these three definitions is a clear winner. 
Rather, different applications may call for different metrics.
From an algorithmic point of view, \ref{def:ooo_dec} and \ref{def:ooo_inc} are essentially identical, in that detecting the out-of-order events only requires comparing adjacent pairs of packets.
An out-of-order event with respect to \ref{def:ooo_max}, however, is far more difficult to uncover, as looking at pairs of packets is no longer enough---the algorithm always has to record the maximum sequence number (over a potentially large number of packets) in order to report such events.
In this paper, we focus on \ref{def:ooo_dec} and show that easy modifications to the algorithms can be effective for \ref{def:ooo_inc}.

\subsubsection{A strawman solution for identifying out-of-order heavy flows} \label{sec:prob:flow:strawman}

A naive algorithm that identifies out-or-order heavy flows would memorize, for every flow, the flow ID $f$, the sequence number $s$ of the latest arriving packet from $f$ when using \ref{def:ooo_dec}, and the number of out-of-order packets $o$.
When a new packet of $f$ arrives, we go to its flow record, and compare its sequence number $s'$ with $s$.
If $s' < s$, the new packet is out-of-order and we increment $o$ by $1$. 

For \ref{def:ooo_inc}, we simply save the expected sequence number $s+1$ of the next packet when maintaining the flow record, and compare it to that of the new packet, according to \ref{def:ooo_inc}.
We see that different definitions only slightly altered the sequence numbers saved in memory, and we always decide whether an out-of-order event has happened based on the comparison.

\subsubsection{Memory lower bound for identifying out-of-order heavy flows} \label{sec:prob:flow:lb}


To show that identifying out-of-order heavy flows is fundamentally expensive, we want to construct a worst-case packet stream, for which detecting heavy reordering requires a lot of memory.
For simplicity, we consider the case where heavy reordering occurs in only one of the $\sizeof{\mathcal{F}}$ flows, and let this flow be $f$.
If $f$ is also heavy in size, it suffices to use a heavy-hitter data structure to identify $f$.
Problems arise when $f$ is not that heavy on any timescale, and yet is not small  enough to be completely irrelevant.
A low-rate, long-lived flow fits such a profile.
Unless given a lot of memory, a heavy-hitter data structure is incapable of identifying $f$.
Moreover, since the packet inter-arrival times for a low-rate flow are large, to see more than one packet from $f$, the record of $f$ would need to remain in memory for a longer duration, relative to other short-lived or high-rate flows.

Next we formalize this intuition, and show that given some flow $f$, it is infeasible for
a streaming algorithm to always distinguish whether $O_f$ is large or not, with memory sublinear in the total number of flows $\sizeof{\mathcal{F}}$, even with randomness and approximation. 
\begin{claim} \label{clm:lb}
Divide a stream with at most $\sizeof{\mathcal{F}}$ flows into $k$ time-blocks $B_1, B_2, \dots, B_k$. It is guaranteed that one of the following two cases holds:
\begin{enumerate}
\item
For any pair of blocks $B_i$ and $B_j$ with $i \neq j$, there does not exist a flow that appears in both $B_i$ and $B_j$.
\item
There exists a unique flow $f$ that appears in $\Theta(k)$ blocks.
\end{enumerate}
Then distinguishing between the the two cases is hard for low-memory algorithms.
Specifically, a streaming algorithm needs $\Omega(\min{(\sizeof{\mathcal{F}}, \frac{\sizeof{\mathcal{F}}}{k}\log{\frac{1}{\delta}}}))$ bits of space to identify $f$ with probability at least $1 - \delta$, if $f$ exists.
\end{claim}

Claim~\ref{clm:lb} follows from reducing the communication problem MostlyDisjoint 
stated in \cite{kamath2021simple}, by treating elements of the sets as flow IDs in a packet stream. 

Claim~\ref{clm:lb} implies the hardness of identifying out-of-order heavy flows, as the unique flow $f$ may have many packets, but not be heavy enough for a heavy-hitter algorithm
to detect it efficiently.
Deciding whether such a flow exists is already difficult, identifying it among other flows is at least as difficult.
Consequently, checking whether it has many out-of-order packets is difficult as well.


The same reduction also implies that detecting duplicated packets
requires $\Omega(\sizeof{\mathcal{F}})$ space.
In fact, Claim~\ref{clm:lb} corroborates the common perception that measuring performance metrics such as round-trip delays, reordering, and retransmission in the data plane is generally challenging, as it is hard to match tuples of packets that span a long period of time, with limited memory.

\subsection{Prefix-level reordering statistics}
\label{subsec:prefix}

\subsubsection{Problem statement} \label{sec:prob:prefix:statement}

Identifying out-of-order heavy flows is hard; fortunately, we do not always need to report individual flows.
Since reordering is typically a property of a network path, and routing decisions are made at the prefix level, it is natural to focus on heavily reordered prefixes.
Throughout this paper, we consider $24$-bit source IP prefixes, as they achieve a reasonable level of granularity.
The same methods apply if prefixes of a different length are more suitable in other applications.

By common definitions of the flow ID, the prefix $g$ of a packet $(f, s, t)$ is encoded in $f$.
To simplify notations, we think of a prefix $g$ as the set of flows with that prefix, and when context is clear, $S$ also refers to the set of all prefixes in the stream.
Let $O_g = \sum_{f \in g} O_f$ be the number of out-of-order packets in prefix $g$.
A prefix $g$ is \emph{out-of-order heavy} if $O_g > \varepsilon N_g$ for some small $\varepsilon > 0$, where $N_g$ is the number of packets in prefix $g$.

For localizing attacks and performance problems, it is not always sensible to catch prefixes with the highest fraction of out-of-order packets.
When a prefix is small, even a single out-of-order packet would lead to a large fraction, but it might just be caused by a transient loss.
In addition, with the control plane being more computationally powerful yet less efficient in packet processing, there is an apparent trade-off between processing speed and the amount of communication from the data plane to the control plane.
As a result, we also want to limit the communication overhead incurred. 

Therefore, for some $\varepsilon, \alpha, \beta$, our goals can be described as: 
\begin{enumerate}
\item Report prefixes $g$ with $N_g \geq \beta$ and $O_{g} > \epsilon N_g$.
\item Avoid reports of prefixes with at most $\alpha$ packets.
\item Keep the communication overhead from the data plane to the control plane small.
\end{enumerate}

\subsubsection{Bypassing memory lower bound}

As a consequence of Claim~\ref{clm:lb}, it is evidently infeasible to study all flows from a prefix and aggregate all of that information to determine whether to report the prefix.
So why would reporting at the prefix level circumvent the lower bound?
In practice, packets are often reordered due to a congested or flaky link that causes lost, reordered, or retransmitted packets at the TCP level. 
Therefore, flows traversing the same path at the same time are positively correlated in their out-of-orderness.
This effectively means that we only need to study a few flows from a prefix to estimate the extent of reordering this prefix suffers.
We state the correlation assumption that all of our algorithms are based on as follows, and postpone its verification to~\S\ref{sec:traffic:correlation}:
\begin{assumption}
Let $f$ be a flow chosen uniformly at random from all flow in prefix $g$.
If $N_g > \alpha$, and $g$ has at least two flows, $\frac{O_g-O_f}{N_g-N_f}$ and $\frac{O_f}{N_f}$ are positively correlated.
\end{assumption}

%% file: measurement.tex
\section{Traffic Characterization}
\label{sec:traffic}

This section presents several traffic traits that drive our algorithm design.
For all of our measurement and evaluation, we make use of the following real-world packet traces: 
\begin{itemize}
    \item \textbf{Campus:} Two anonymized packet traces, collected ethically from a border router on a university campus network on June 5, 2019, and May 9, 2022, respectively.
    \item \textbf{Backbone:} CAIDA Anonymized Internet Traces from 2018~\cite{caida18} and 2019~\cite{caida19}.
\end{itemize}

Note that only packets with payloads are relevant for our application, as TCP sequence numbers must advance for our algorithms to detect reordering events.
We therefore preprocess the trace to only contain flows from servers to clients using source and destination port numbers,
with the rationale that these senders are more likely to generate continuous streams of traffic.



\subsection{Heavy-tailed size and out-of-orderness} \label{sec:traffic:dist}


Consistent with numerous prior measurement studies, in our $5$-minute campus trace (Figure~\ref{fig:pref_flow_cdf}), most flows are small, and only a few flows are large.
However, a small fraction of flows and prefixes tend to account for a large fraction of the traffic.
For example, in this trace, more than $90\%$ of the packets belong to  the $5\%$ largest flows or prefixes.
Out-of-orderness in prefixes is similarly heavy-tailed; only a small fraction of prefixes have a significant fraction of packets out-of-order (Figure~\ref{fig:ooo_cdf}), 
e.g., only less than $12\%$ of prefixes with 
at least $2^7$ packets have more than $1\%$ of packets out of order by \ref{def:ooo_dec}.
Out-of-order events defined by \ref{def:ooo_inc} are more prevalent.
But, even so, packet reordering remains a low-probability event, with less than $10\%$ of the prefixes of size at least $2^7$ experiencing more than $7\%$ out-of-order packets by \ref{def:ooo_inc}.


If most packet reordering occurred in heavy flows and prefixes, detecting heavy reordering would be easy, by solely focusing on large flows and prefixes using heavy-hitter data structures.
However, what happens in reality is quite the opposite.
To see that, we use an unconventional split violin plot (Figure~\ref{fig:violin}) to show three sets of information: the prefix size (color of the violin), the flow size distribution in a prefix (the left half of the violin), and the fraction of reordered packets for that prefix that lie within flows of certain size (the right half of the violin).
Each split violin corresponds to a heavily reordered prefix with at least $\beta = 2^7$ packets, using \ref{def:ooo_inc} with $\varepsilon = 0.02$.
By comparing the left halves of all violins, we see a wide variation of flow sizes in prefixes with heavy reordering, and the sizes of such prefixes can be orders-of-magnitude different.
We see that many prefixes do not have any large flows.
Moreover, the largest flows in each heavily reordered prefix do not necessarily contain most of the out-of-order packets in that prefix.
The $131$-largest prefix gives one such example. 
Though the $50$ largest flows in this prefix have size $2^7$ or larger, almost $95\%$ of the total out-of-order packets in this prefix comes from flows with size smaller than $2^7$.
Such a prefix would be very difficult for a heavy-hitter data structure to catch without investing significant memory.
Thus, by zooming in on large flows and prefixes, we would inevitably miss out on many prefixes of interest without any large flow.

\begin{figure*}[t]
\centering
\begin{subfigure}[t]{0.32\linewidth}
\centering
\includegraphics[width=\textwidth]{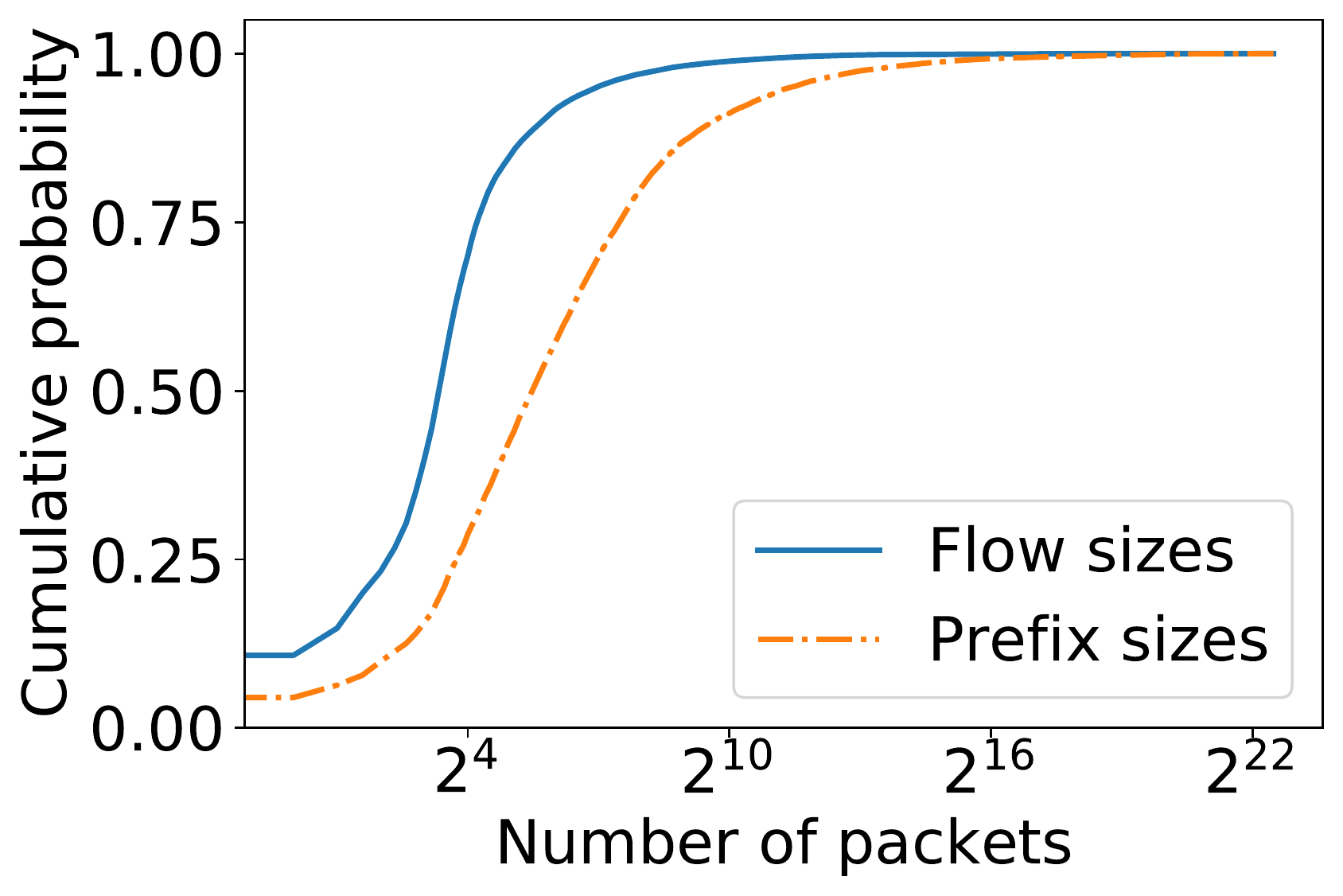}
\caption{
A small fraction of flows and prefixes account for a large fraction of the traffic.}
\label{fig:pref_flow_cdf}
\end{subfigure} \hfill
\begin{subfigure}[t]{0.32\linewidth}
\centering
\includegraphics[width=\textwidth]{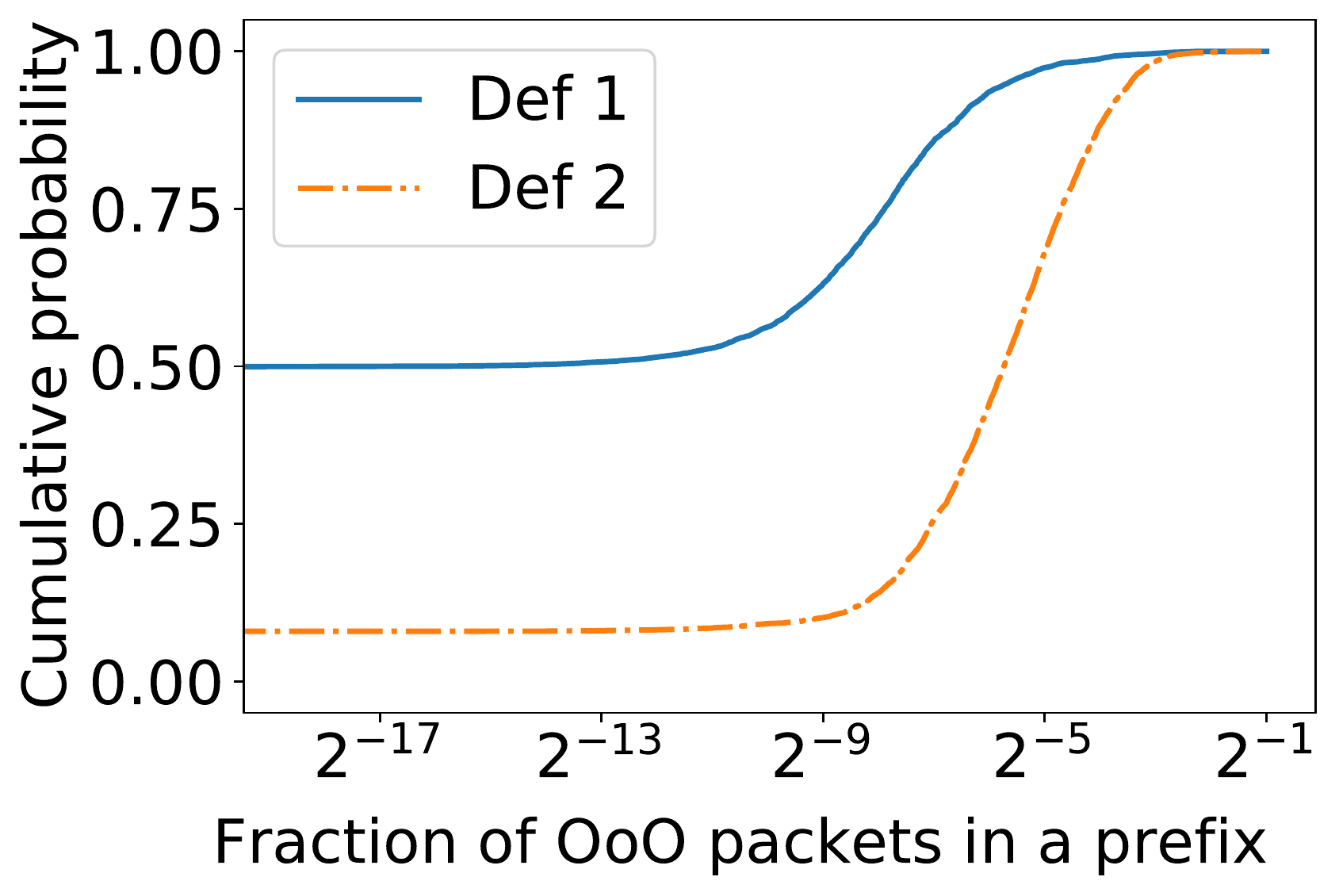}
\caption{Out-of-order heavy prefixes are rare.
Here prefixes have at least $\beta = 2^7$ packets.}
\label{fig:ooo_cdf}
\end{subfigure} \hfill
\begin{subfigure}[t]{0.32\linewidth}
\centering
\includegraphics[width=\textwidth]{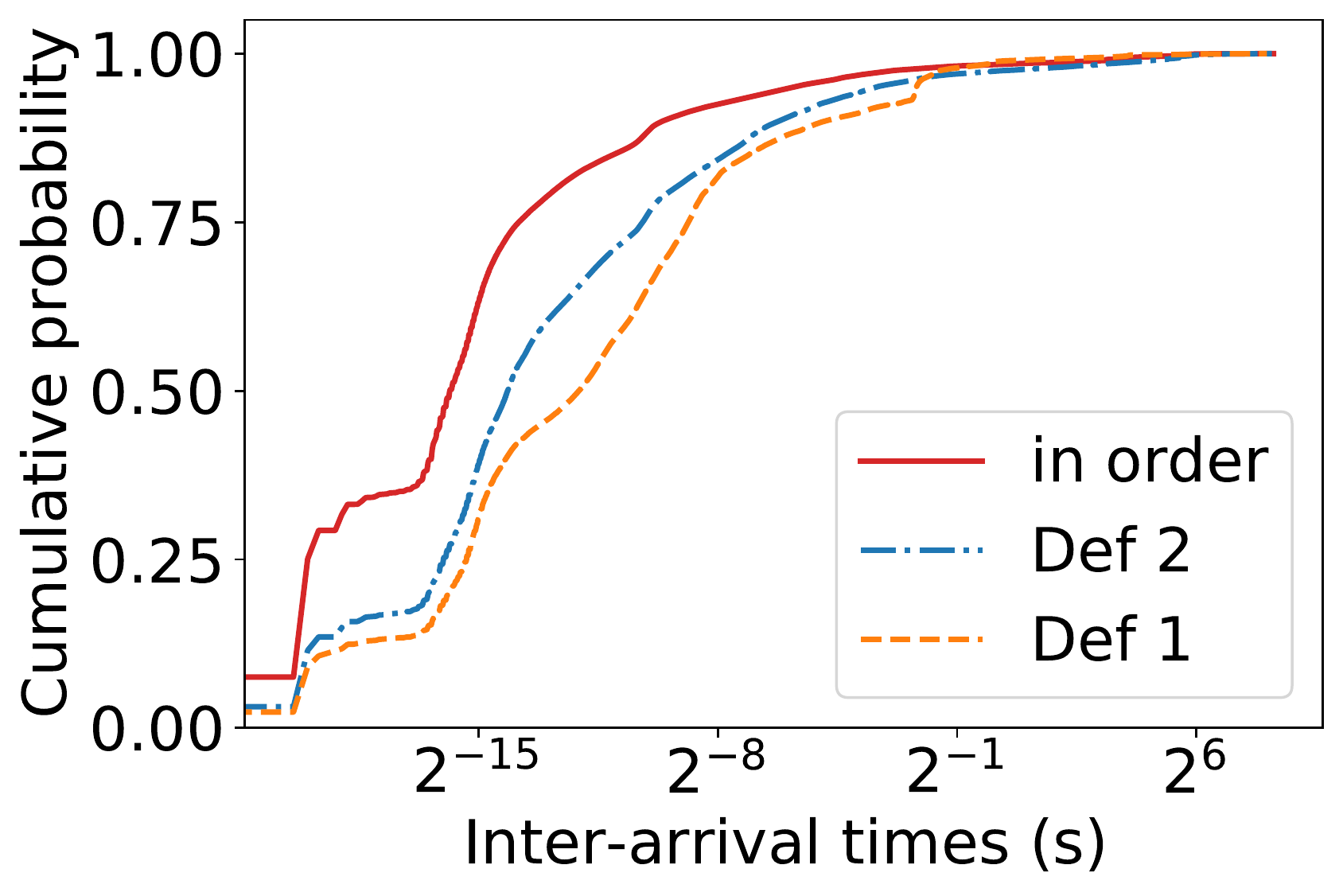}
\caption{
Out-of-order events defined by \ref{def:ooo_dec} exhibit the highest inter-arrival times.
}
\label{fig:inter}
\end{subfigure}
\vspace{-0.9em}
\caption{Heavy-tailed distributions in a $5$-minute campus trace.}
\vspace{-1em}
\end{figure*}


\begin{figure*}[t]
\centering
\includegraphics[width=1\textwidth]{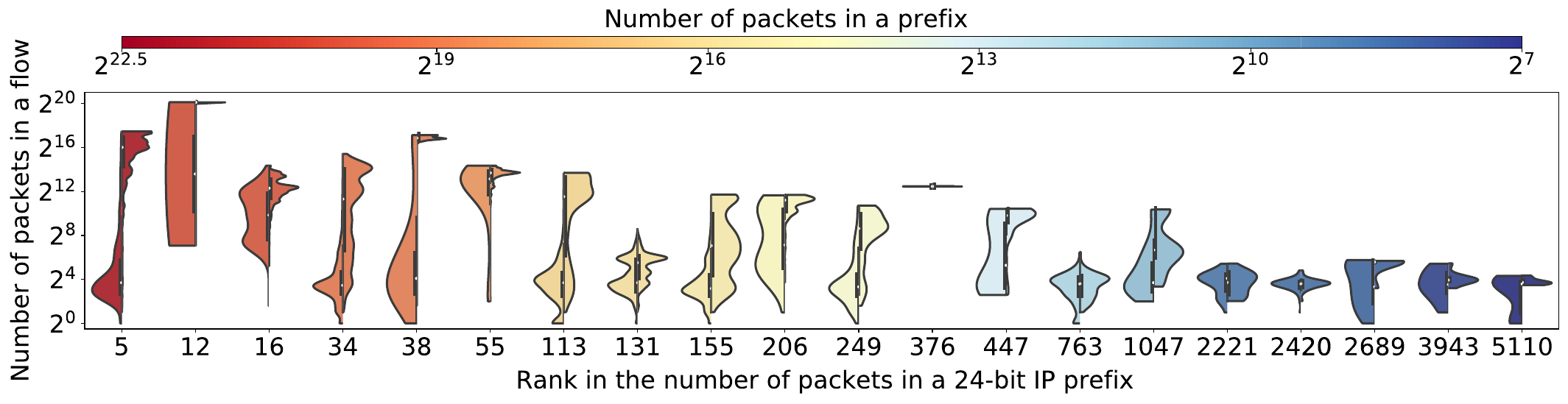}
\vspace{-2em}
\caption{
A split violin plot showing prefix sizes, distributions of flow sizes in each prefix, and what fraction of reordering in a prefix comes for which flow size. 
A split violin of rank $r$ refers to the $r$-th largest prefix in the trace.
}
\label{fig:violin}
\vspace{-1em}
\end{figure*}

Fortunately, to report a prefix with a significant amount of reordering, we need not measure every flow in that prefix, as flows in the same prefix have some correlation in their out-of-orderness.
As it turns out, the fraction of out-of-order packets in a prefix is positively correlated with that of a flow within the prefix, which we verify next.

\subsection{Correlation among flows in a prefix} \label{sec:traffic:correlation}

Let $f$ be a flow drawn uniformly at random from a set of flows.
Let $X$ be the random variable representing the fraction of out-of-order packets in flow $f$, $X = \frac{O_f}{N_f}$.
Denote $g$ as the prefix of flow $f$, let $Y$ be the random variable denoting the fraction of out-of-order packets among all flows in prefix $g$ excluding $f$, that is, $Y = \frac{O_g - O_f}{N_g - N_f}$, where $N_g$ is the number of packets in prefix $g$.
To ensure that $N_g > N_f$, the prefixes we sample from must have at least two flows.
We use the \emph{Pearson correlation coefficient (PCC)} to show that $X$ and $Y$ are positively correlated, which implies that the out-of-orderness of a flow $f$ is statistically representative of other flows in the prefix of $f$.
Essentially a normalized version of $\Cov(X,Y)$, PCC always lies in the interval $[-1, 1]$, and a positive PCC indicates a positive linear correlation.
Lacking a better reason to believe the correlation between $X$ and $Y$ is of higher order, we shall see that PCC suffices for our analysis.


Given a traffic trace, let $S$ be the set of flows whose prefixes 
have at least two flows.
We compute the PCC as follows:
\begin{enumerate}
    \item 
    Draw $n$ flows from $S$, independently and uniformly at random.
    \item 
    For each of the $n$ flows $f_i$, let $x_i = \frac{O_{f_i}}{N_{f_i}}$, $y_i = \frac{\sum_{f' \in g, f' \neq f_i} O_{f'}}{\sum_{f' \in g, f' \neq f_i} N_{f'}}$,
    \item
    The PCC $r = \frac{\sum_{i=1}^n (x_i - \bar{x})(y_i - \bar{y})}{\sqrt{\sum_{i=1}^n (x_i - \bar{x})^2} \sqrt{\sum_{i=1}^n (y_i - \bar{y})^2}}$, where $\bar{x} = \frac{1}{n} \sum_{i=1}^n x_i, \bar{y} = \frac{1}{n} \sum_{i=1}^n y_i$.
\end{enumerate}

\begin{figure}[t]
\centering
\includegraphics[width=0.47\textwidth]{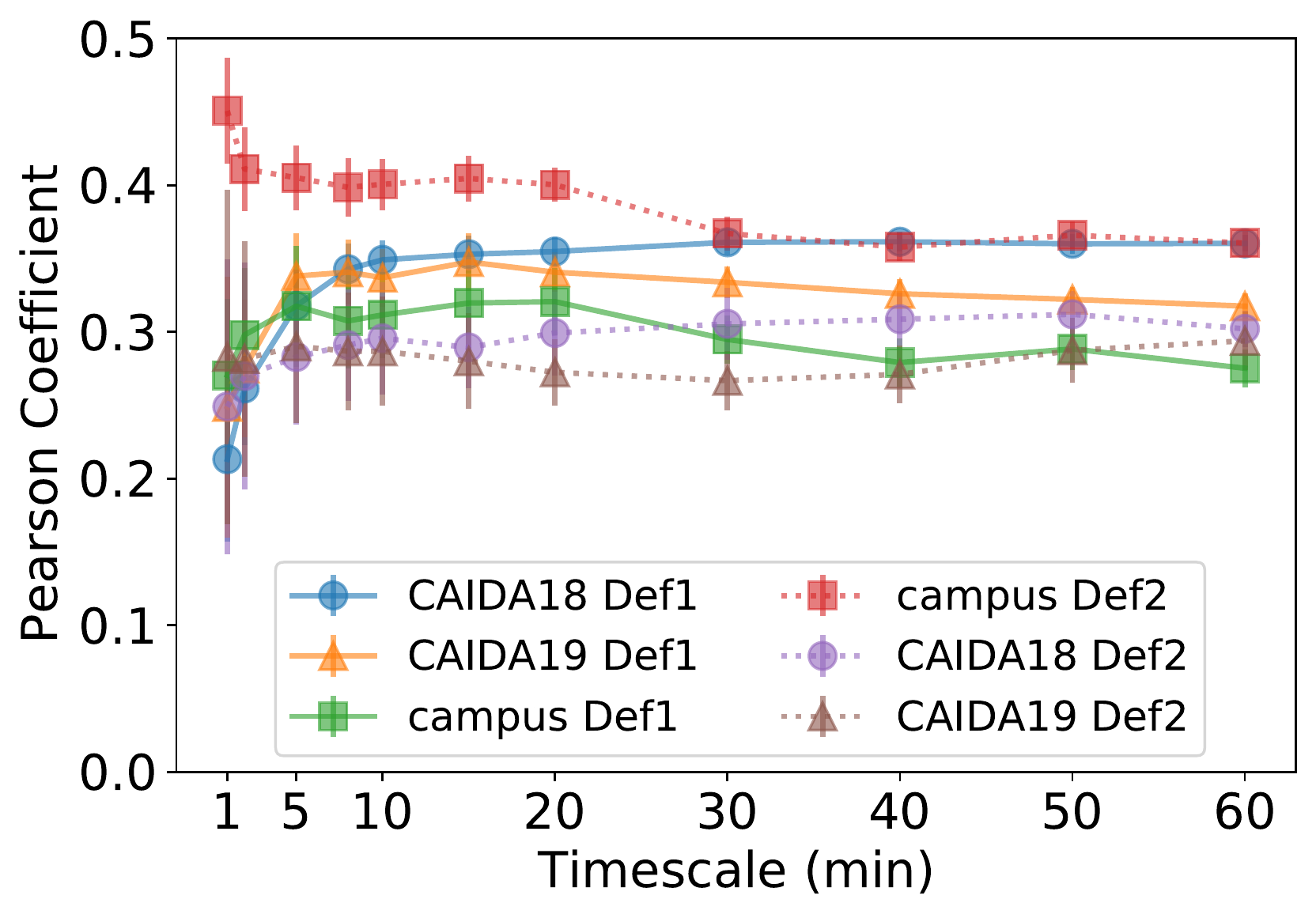}
\vspace{-0.8em}
\caption{Pearson coefficient on varying timescales shows that a positive correlation exists between the reordering of a flow and that of its prefix.
}
\label{fig:pcc}
\vspace{-1em}
\end{figure}

We perform $m=100$ tests on each traffic trace using both definitions of reordering, on timescales ranging from $1$ minute to $60$ minutes (Figure~\ref{fig:pcc}). 
Each point shows the average of $m=100$ tests, where we draw $n = 0.5\% \cdot \sizeof{S}$ flows in each test.
The result indicates that, a positive correlation exists between $X$ and $Y$ for all tested traces on all timescales, and the correlation tends to stabilize after a small time period such as five minutes.

\subsection{Packet inter-arrival times within a flow} \label{sec:traffic:inter}

We also study the inter-arrival time of packets within a flow to understand how efficient the flow sampling algorithm can be.
Due to TCP windowing dynamics, where the sender transmits a window of data and then waits for acknowledgments, in-order packets tend to have small inter-arrival times.
Depending on the definition, reordering can be a result of gaps in transmission of non-consecutive packets (\ref{def:ooo_inc}), or worse yet the retransmissions of lost packets (\ref{def:ooo_dec}), which often lead to larger inter-arrival times.

Indeed, Figure~\ref{fig:inter} shows that the inter-arrival times of out-of-order packets using \ref{def:ooo_inc} 
tend to be smaller than that of the out-of-order packets using \ref{def:ooo_dec}, with the inter-arrival times of in-order packets being the smallest.
This implies that, to detect the reordering events in \ref{def:ooo_dec}, the algorithm has to store  records for a longer waiting period, which potentially exhausts more memory resources.


%% file: algo.tex
\section{Data-Plane Data Structures for Out-of-Order Monitoring} \label{sec:algo}

At a high level, a data-plane algorithm generates reports of flows with potentially heavy packet reordering on the fly, and a simple control-plane program  parses through the reports to extract their prefixes.
Each report includes the prefix, the number of packets monitored, and the number of out-of-order packets of a suspicious flow.
At the end of the time interval, we can also scan the data-plane data structure to generate reports for highly-reordered flows remaining in memory.
On seeing reports, a control-plane program simply aggregates counts from reports of the same prefix, and outputs a prefix when its count exceeds a threshold.

In the data plane, we keep state at the flow level, and consider prefix information in allocating memory.
Assuming a positive correlation between the out-of-orderness of a prefix and that of the flows from that prefix, we do not have to monitor all flows in their entirety to gain enough information about a prefix.
This leads to the simple yet effective flow-sampling algorithm in \S~\ref{sec:algo:array}, where we sample as many flows as possible, but only over a short period at a time.
Though it is not enough to only measure reordering in heavy flows (\S~\ref{sec:traffic:dist}), in \S~\ref{sec:algo:hybrid}, we show that there are still benefits from combining a heavy-hitter data structure with the flow-sampling array.

\subsection{Sample flows over short periods}
\label{sec:algo:array}

To sieve through a large number of flows with limited memory, the turnover rate has to be high.
This means that, the algorithm has to be somewhat oblivious to the various statistics of a flow, such as flow sizes and inter-arrival times, when choosing to admit or evict a flow.
To set the stage for later discussions, throughout this paper, we refer to the unit of memory allocated to keep one flow record as a \emph{bucket}.
Now, rather than one bucket per flow, the main idea is to use one bucket to quickly check over multiple flows in turn.

\subsubsection{Flow sampling with array} \label{sec:algo:array:original}

Under the strict memory access constraints, we again opt for a hash-indexed array as a natural choice of data structure, where each row in the array corresponds to a bucket, and all buckets behave independently.
To check many prefixes for reordering, we do not want some prefix with a huge number of flows to dominate the data structure.
To this end, we assign flows from the same prefix to the same bucket, by hashing prefixes instead of flow IDs, a technique we use in all our algorithms.

Therefore, we fix a bucket $\mathfrak{b}$, and consider the substream of packets hashed to $\mathfrak{b}$.
When a packet $(f, s, t)$ arrives at $\mathfrak{b}$, there are three cases:
\begin{enumerate}
    \item
    If $\mathfrak{b}$ is empty, we always admit the packet, that is, we save its flow record $f$, sequence number $s$, timestamp $t$ in $\mathfrak{b}$, together with the number of packets $n$ and the number of out-of-oder packets $o$, both initilized to $0$.
    \item
    If flow $f$'s record is already in $\mathfrak{b}$, we update the record as in the strawman solution (\S~\ref{sec:prob:flow:strawman}), and update the timestamp in memory to $t$.
    \item
    If $\mathfrak{b}$ is occupied by another flow's record $(f', s', t', n', o')$, we only admit $f$ if $f'$ has been monitored in memory for a sufficient period specified by parameters $T$ and $C$, or the prefix of $f'$ could be potentially heavily reordered with respect to another parameter $R$.
    That is, $f$ overwrites $f'$ with record $(f, s, t, n=0, s=0)$ only if one of the following holds:
    \begin{enumerate}
        \item $f'$ is stale: $t - t' > T$.
        \item $f'$ has been hogging $\mathfrak{b}$ for too long: $n' > C$.
        \item \label{algo:report}
        $f'$ might belong to a prefix with heavy reordering: $o' > R$. 
    \end{enumerate}
\end{enumerate}

In Case~\ref{algo:report}, the algorithm sends a $3$-tuple report $(g', n', o')$ to the control plane, where $g'$ is the prefix of flow $f'$.
On seeing reports from the data plane, a simple control-plane program keeps a tally for each reported prefix $g$.
Let $\{(g, n_i, o_i)\}_{i=1}^r$ be the set of all reports corresponding to a prefix $g$.
The control-plane program outputs $g$ if $\sum_{i=1}^r n_i \geq \alpha$, for the same $\alpha$ in \S~\ref{sec:prob:prefix:statement}.
In the following sections, we refer to the data-plane component together with the simple control-plane program as the flow-sampling algorithm. 

\paragraph{Lazy expiration of flow records in memory}

Due to memory access constraints, many data-plane algorithms \emph{lazily} expire records in memory on collisions with other flows, as opposed to actively searching for stale records in the data structure.
We again adopt the same technique in the algorithm above, though here it is more nuanced.
We could imagine a variant of the algorithm where a flow is monitored for up to $C+1$ packets at a time.
That is, when the $(C+1)$st packet arrives, we check whether to report this flow, and evict its record.
Compared to this variant, lazy expiration helps in preventing a heavy flow being admitted into the data structure consecutively, so that the heavy flow can be evicted before a integer multiple of $(C+1)$ packets, should another flow appear in the meantime.

\paragraph{Robustness of flow sampling}

For the flow-sampling method to be effective, the data structure needs to sample as many flows as possible.
Therefore, it is not desirable to keep a large flow in memory when we have already seen many of its packets, and learned enough information about its packet reordering.
This means the packet count threshold $C$ should not be too large.
Neither do we want to keep a flow, regardless of its size, that has long been finished.
We can eliminate such cases by setting a small inter-arrival timeout $T$.

Now the question is, how small can these parameters be.
Real-world traffic can be bursty, meaning that sometimes there are packets from the same flow arriving back-to-back.
In this case, even if we overwrite the existing flow record on every hash collision ($T=0$ and $C=1$), the algorithm still generates meaningful samples.
When the memory is not too small compared to the number of prefixes, and hash collisions are rare, the algorithm might even have good performance.
However, setting small $T>0$ and $C>1$ makes the algorithm more robust against worst-case streams.
Consider a stream of packets where no adjacent pairs of packets come from the same flow.
On seeing such a stream, a flow-sampling algorithm that overwrites existing records on every hash collision with another flow will no doubt collect negligible samples.
In contrast, small $T>0$ and $C>1$ allow a small period of time for a flow in memory to be monitored, and hence gives a better chance of capturing packet reordering.

\subsubsection{Performance guarantee}
In this section, we analyze the number of times a flow with a certain size is sampled.
Consider a prefix $g$ when the hash function is fixed.
Let $\mathfrak{b}$ be the bucket prefix $g$ is hashed to, and we know all the flows as well as the prefixes that are hashed to $\mathfrak{b}$.
With a slight abuse of notation, we write $g \in \mathfrak{b}$ when the bucket with index $h(g)$ is $\mathfrak{b}$.
We also write $f \in \mathfrak{b}$ when $f$'s prefix is hashed to $\mathfrak{b}$.
To capture the essence of the flow-sampling algorithm
without excessive details, we make the following assumptions:
\begin{enumerate}
\item 
Each packet in $S$ is sampled i.i.d. from distribution $(p_{f})_{f \in \mathcal{F}}$, that is, each packet belongs to some flow $f \in \mathcal{F}$ independently with probability $p_f$.
Consequently, each packet belongs to some prefix $g$ independently with probability $p_g = \sum_{f \in g} p_f$.
\item \label{assump:T}
Let $p_{f \mid \mathfrak{b}} = \frac{p_{f}}{\sum_{f' \in \mathfrak{b}} p_{f'}}$, $p_{g \mid \mathfrak{b}}$ can be similarly defined.
Only a flow $f$ with $p_{f \mid \mathfrak{b}}$ greater than some $p_{\text{min}}$ will get checked, where we think of $p_{\text{min}}$ as a fixed threshold depending on the inter-arrival time threshold $T$ and distribution $(p_{f})_{f \in \mathcal{F}}$.
\item \label{assump:C}
A flow is checked exactly $C+1$ packets at a time.
\end{enumerate}
Note that Assumption (\ref{assump:T}) is a way to approximate the effect of $T$, where we assume a low-frequency flow would soon be overwritten by some other flow on hash collision.
In contrast to Assumption (\ref{assump:C}), the flow sampling algorithm does \emph{not} immediately evict a flow record with $C+1$ packets, if there is no hash collision.
In this way, though $f$ is monitored beyond its original $C+1$ packets, once a hash collision occurs, the collided flow would seize $f$'s bucket.
By imposing Assumption (\ref{assump:C}), the heavier flows would likely benefit by getting more checks, while the smaller flows would likely suffer.
Empirically, the eviction scheme of the flow-sampling algorithm (\S~\ref{sec:algo:array:original}) achieves better performance in comparison to Assumption (\ref{assump:C}). \vspace{-0.2em}
\begin{lemma} \label{lem}
Given the total length of stream $\sizeof{S}$, distributions $(p_{f})_{f \in \mathcal{F}}$, 
with the assumptions above, for a fixed hash function $h$ and any $\varepsilon, \delta \in (0,1)$, a prefix $g$ in bucket $\mathfrak{b}$ is checked at least $(1- \delta) t_1 p_{g \mid \mathfrak{b}}$ times with probability at least $1- e^{-p_{\text{min}} t_1 C F_{\mathfrak{b}} \cdot \frac{\varepsilon^2}{24}} - e^{-\frac{\varepsilon^2 \sizeof{S} \sum_{g \in \mathfrak{b}} p_g}{3}} - e^{-\frac{\delta^2 t_1 p_{g \mid \mathfrak{b}}}{2}}$, where $t_1 = \floor{\frac{\sizeof{S} \sum_{g \in \mathfrak{b}} p_g}{(1+\frac{\varepsilon}{2}) C F_{\mathfrak{b}}}}$ and $p_{g \mid \mathfrak{b}} = \frac{\sum_{f \in g: p_{f \mid \mathfrak{b}} \geq p_{\text{min}}} p_f}{\sum_{f' \in \mathfrak{b}} p_{f'}}$.
\end{lemma}
\begin{proof}
Let $S_{\mathfrak{b}}$ the substream of $S$ that is hashed to $\mathfrak{b}$.
Given $\sizeof{S}$, the length $\sizeof{S_{\mathfrak{b}}}$ of substream $S_{\mathfrak{b}}$ is a random variable, $\E \sizeof{S_{\mathfrak{b}}} = \sizeof{S} \sum_{g \in \mathfrak{b}} p_g$, then by Chernoff bound,
\begin{equation} \label{ineq:upper_sizeof_S_b}
	\Pr[\sizeof{S_{\mathfrak{b}}} < (1-\varepsilon) \E \sizeof{S_{\mathfrak{b}}}]
	< e^{-\frac{\varepsilon^2 \E \sizeof{S_{\mathfrak{b}}}}{3}}
	= e^{-\frac{\varepsilon^2 \sizeof{S} \sum_{g \in \mathfrak{b}} p_g}{3}}.
\end{equation}

Let $t$ be a random variable denoting the number of checks in $\mathfrak{b}$.
Let random variable $X_{i,j}$ be the number of packets hashed to $\mathfrak{b}$ after seeing the $j$th packet till receiving the $(j+1)$st packet from the currently monitored flow, where $i \in [t]$ and $j \in [C]$.
$X_{i,j}$s are independent geometric random variables, and $X_{i,j} \sim Geo(p_{f_i \mid b})$, where $f_i$ is the flow under scrutiny during the $i$th check, by Assumption~\ref{assump:T}, $p_{f_i \mid b} \geq p_{\text{min}}$.
Next we look at $X = \sum_{i=1}^t \sum_{j=1}^C X_{i,j}$, the length of the substream in $\mathfrak{b}$ after $t$ checks,
\begin{align} \label{eqn:EX}
	\E X
	= \sum_{i=1}^t \sum_{j=1}^C \E X_{i,j}
	= \sum_{i=1}^t \sum_{j=1}^C \sum_{\substack{f \in \mathfrak{b}: \\ p_{f \mid b} \geq p_{\text{min}}}} p_{f \mid b} \cdot \frac{C}{p_{f \mid b}}
	= t C F_{\mathfrak{b}},
\end{align}
where $F_{\mathfrak{b}} = \sizeof{\{f \in b \mid p_{f \mid b} \geq p_{\text{min}} \}}$.
By the Chernoff-type tail bound for independent geometric random variables (Theorem 2.1 in~\cite{janson2018tail}), for any $\varepsilon \in (0,1)$,
\begin{align} \label{ineq:upper_X}
	\Pr[X > (1+\frac{\varepsilon}{2}) \E X]
	< e^{-p_{\text{min}} \E X (\frac{\varepsilon}{2} - \ln{(1+\frac{\varepsilon}{2})})}
	\leq e^{-p_{\text{min}} t C F_{\mathfrak{b}} \cdot \frac{\varepsilon^2}{24}}.
\end{align}

Let $t_1$ be the largest $t$ such that $(1+\frac{\varepsilon}{2}) \E X < \E \sizeof{S_{\mathfrak{b}}}$, we have $t_1 = \floor{\frac{\sizeof{S} \sum_{g \in \mathfrak{b}} p_g}{(1+\frac{\varepsilon}{2}) C F_{\mathfrak{b}}}}$.
Consider two events:
\begin{enumerate}[(i)]
\item
The number of checks $t$ on seeing $S_{\mathfrak{b}}$ is less than $t_1$.

Applying~\ref{ineq:upper_X} on $t_1$, we have that with probability at most $e^{-p_{\text{min}} t_1 C F_{\mathfrak{b}} \cdot \frac{\varepsilon^2}{24}}$, after seeing $(1 - \varepsilon) \E \sizeof{S_b}$ packets, the number of checks is at most $t_1$.
Together with~\ref{ineq:upper_sizeof_S_b}, by union bound, 
\begin{equation} \label{ineq:event1}
	\Pr[t < t_1] 
	< e^{-p_{\text{min}} t_1 C F_{\mathfrak{b}} \cdot \frac{\varepsilon^2}{24}} + e^{-\frac{\varepsilon^2 \sizeof{S} \sum_{g \in \mathfrak{b}} p_g}{3}}.
\end{equation}

\item 
Prefix $g$ is checked less than $(1- \delta) t_1 p_{g \mid \mathfrak{b}}$ times.
By Chernoff bound, this event holds with probability at most $e^{-\frac{\delta^2 t_1 p_{g \mid \mathfrak{b}}}{2}}$.
\end{enumerate}
The Lemma follows from applying the
union bound over these two events.
\end{proof}

Counterintuitively, the proof of Lemma~\ref{lem} suggests hash collisions are in fact harmless in the flow-sampling algorithm, for a flow that is not too small (which corresponds to $p_{f \mid \mathfrak{b}}$ greater than some $p_{\text{min}}$ in Assumption~\ref{assump:T}).
To see that, suppose we add another heavy flow to bucket $\mathfrak{b}$, $\E \sizeof{S_{\mathfrak{b}}}$ would increase by some factor $x$, which means $\E X$ would increase by the same factor.
Since $F_{\mathfrak{b}}$ would only increase by $1$, if $F_{\mathfrak{b}}$ is large enough, by~\eqref{eqn:EX}, $t$ would also increase by roughly a factor of $x$, while $p_{f \mid \mathfrak{b}}$ decreases by roughly a factor of $x$.
Then $t \cdot p_{f \mid \mathfrak{b}}$ is about the same with or without the added heavy flow.
Therefore, colliding with heavy flows does not decrease the number of checks of a flow that is not too small, as long as the total number of flows in a bucket is large enough, which is usually the case in practice.

\subsubsection{Decrease the number of false positives} \label{sec:algo:array:all}

Since the parameters of the flow-sampling algorithm are chosen so that many flows are sampled, and some might get sampled multiple times, 
it is possible for the algorithm to capture many out-of-order events, but not every one of them indicates that the prefix is out-of-order heavy.
After all, there is only a weak correlation between the out-of-orderness of flows and that of their prefixes, not to mention that even if the correlation is stronger, we are inferring the extent of reordering on a scale much larger than the snippets of flows that we observe.
In such cases, the algorithm could output many false positives.

To reduce the number of false positives, we could imagine feeding the control plane more information, so that the algorithm can make a more informed decision about whether the fraction of out-of-order packets exceeds $\varepsilon$, for each reported prefix.
To this end, we modify the flow-sampling algorithm 
to always report before eviction, even if the number of out-of-order packets is below threshold $R$.
Again denote $\{(g, n_i, o_i)\}_{i=1}^r$ as the set of all reports corresponding to a prefix $g$, the control plane outputs $g$ if $\sum_{i=1}^r n_i \geq \alpha$, and $\frac{\sum_{i=1}^r o_i}{\sum_{i=1}^r n_i} > c \cdot \varepsilon$, for some tunable parameter $0<c \leq 1$.
The parameter $c$ compensates for the fact that we only monitor a subset of the traffic, so the exact fraction of out-of-order packets we observe might not directly align with $\varepsilon$.

\subsection{Separate large flows} \label{sec:algo:hybrid}

Though hash collisions generally do not affect the flow-sampling algorithm's ability to check flows that are not too small, there is still the possibility that a small flow just so happens to arrive and finish during the short period when another flow is being monitored in that bucket.
Such a small flow would never get a second chance to enter the data structure.
If we could instead continuously monitor some large flows in a separate data structure, then for a small flow $f$ that is hashed to a bucket $\mathfrak{b}$ that no longer contains large flows, $p_{f \mid \mathfrak{b}}$ would increase, which would increase the number of checks it gets.
For some prefixes whose out-of-order packets concentrate only in one small flow, separating large flows greatly improves the chance of catching them.

Therefore, we propose a hybrid scheme, where the packets first go through a heavy-hitter (HH) data structure, and the array only admits flows that are not being monitored in the HH data structure. 
We again assign flows with the same prefix to the same set of buckets, and the array part of the data structure behaves exactly as depicted in \S~\ref{sec:algo:array:original}. 
For the HH part, we report flows whose fraction of out-of-order packets is above $\varepsilon$.
For the specifics on the HH data structure, we refer the readers to~\ref{sec:precision}.

Note that a subtly different design choice would be to have the array admit the set of flows whose prefixes are not being monitored in the HH data structure.
This would have made more sense, if all the heavily reordered prefixes have most of their out-of-order packets concentrated among the heaviest flows in that prefix.
But as we have seen in Figure~\ref{fig:violin}, this is not always the case.
Compared to our proposed hybrid algorithm, this variant would be less accurate.
However, it certainly reduces the number of false positives and the number of reports generated by the data-plane algorithm, since in this case, a much smaller set of flows would be monitored by the array.
In this work, we choose to prioritize accuracy over other aspects, so we prefer the hybrid algorithm in last paragraph to this variant.

In any practical setting, the correct memory allocation between the HH data structure and the array in the hybrid scheme depends on the workload properties: the relationship of flows to prefixes, the heaviness of flows and prefixes, and where the reordering actually occurs.
Next we understand how these algorithms behave under real-world workloads.

%% file: eval.tex
\section{Evaluation} \label{sec:eval}

We start this section by evaluating our flow-sampling algorithm and hybrid scheme (\S~\ref{sec:eval:compare}) using a Python simulator on real-world traces introduced in \S~\ref{sec:traffic}.
As much as we wish that each trace is representative, we cannot simply assume that every network administrator running our algorithms in their networks would get the exact same performance.
Therefore, we delve into the intricacies of multiple distributions underlying the real-world traffic workload, to explain how they affect the performance of our algorithms. 
In \S~\ref{sec:eval:lucid}, we verify that our P4 prototype of the flow-sampling algorithm for the Tofino1 switch only consumes a small amount of hardware resources, as promised.
Finally, we recognize that the optimal parameters for our algorithms are often workload dependent.
Thus, we do not attempt to always find the optimum; instead, we show in \S~\ref{sec:eval:compare} that \emph{reasonably} chosen parameters already give good performance.
In \S~\ref{sec:eval:params}, we see that the parameters we used previously for evaluations are indeed representative, and the algorithms are robust against small perturbations.

\subsection{Performance comparisons} \label{sec:eval:compare}

\subsubsection{Metrics}

We begin by introducing the three metrics we use throughout this section to evaluate our algorithms.
Let $\hat{G}$ denote the set of prefixes output by an algorithm $\mathcal{A}$.
\begin{itemize}
\item \textbf{Accuracy}:
Let $G_{\geq \beta} = \{g^* \in S \mid N_{g^*} \geq \beta, O_{g^*} > \varepsilon \sum_{g \in S} O_g \}$ be the ground truth set of heavily reordered prefixes with at least $\beta$ packets.
Define the \emph{accuracy} $A$ of algorithm $\mathcal{A}$ to be the fraction of ground-truth prefixes output by $\mathcal{A}$, that is, 
\[
	A(\mathcal{A}) = \frac{\sizeof{\hat{G} \cap G_{\geq \beta}}}{\sizeof{G_{\geq \beta}}}.
\] 
\item \textbf{False-positive rate}:
Let $G_{> \alpha} = \{g^* \in S \mid N_{g^*} > \alpha, O_{g^*} > \varepsilon \sum_{g \in S} O_g \}$, 
then the \emph{false-positive rate} of $\mathcal{A}$ is defined as
\[
	FP(\mathcal{A}) = \frac{\sizeof{\hat{G} \setminus G_{\geq \alpha}}}{\sizeof{G_{\geq \alpha}}}.
\]
\item \textbf{Communication overhead}:
The \emph{communication overhead} from the data plane to the control plane is defined as the number of reports sent by $\mathcal{A}$, divided by the length of stream $S$, where the number of reports also accounts for the
flow records in the data structure that exceed the reporting thresholds.
\end{itemize}

Unless otherwise specified, each experiment is repeated five times with different seeds to the hash functions, and with parameters $T = 2^{-15}, C = 2^4$, $R_{\text{array}} = 1$, $R_{\text{HH}} = 0.01$, and $d_{\text{HH}}=2$ (see \S~\ref{sec:precision} for details on the parameters of the HH data structure).
We are interested in identifying prefixes with at least $\beta=2^7$ packets, with more than $\varepsilon = 0.01$ fraction of their packets reordered.
Additionally, we do not wish to output prefixes with at most $\alpha = 2^4$ packets, irrespective of their out-of-orderness.

\subsubsection{Performance evaluation} \label{sec:eval:compare:eval}

To the best of our knowledge, we are the first to consider the problem of detecting heavily reordered prefixes, and existing related works are not directly comparable. 
We therefore compare our proposed algorithms to a heavy-hitter (HH) data structure that tracks reordering (\S~\ref{sec:precision}).
Figure~\ref{fig:arr_hybrid_hh} shows the performance of the flow-sampling algorithm, the hybrid scheme, and the HH data structure using \ref{def:ooo_dec}, on a $5$-minute campus trace consisting of $82,359,405$ server-to-client packets, which come from $545,973$ flows and $16,988$ 24-bit source IP prefixes. 
In fact, the specific length of the trace, and whether we choose to study reordering events of \ref{def:ooo_dec} or \ref{def:ooo_inc}, do not affect the overall trend of these curves. 
Due to space limitations, we show the same evaluation (Figure~\ref{fig:arr_hybrid_hh_caida}) on a 10-minute CAIDA 2019~\cite{caida19} trace using \ref{def:ooo_inc} in \S~\ref{sec:supp_eval}.

\begin{figure*}[t]
\centering
\includegraphics[width=\textwidth]{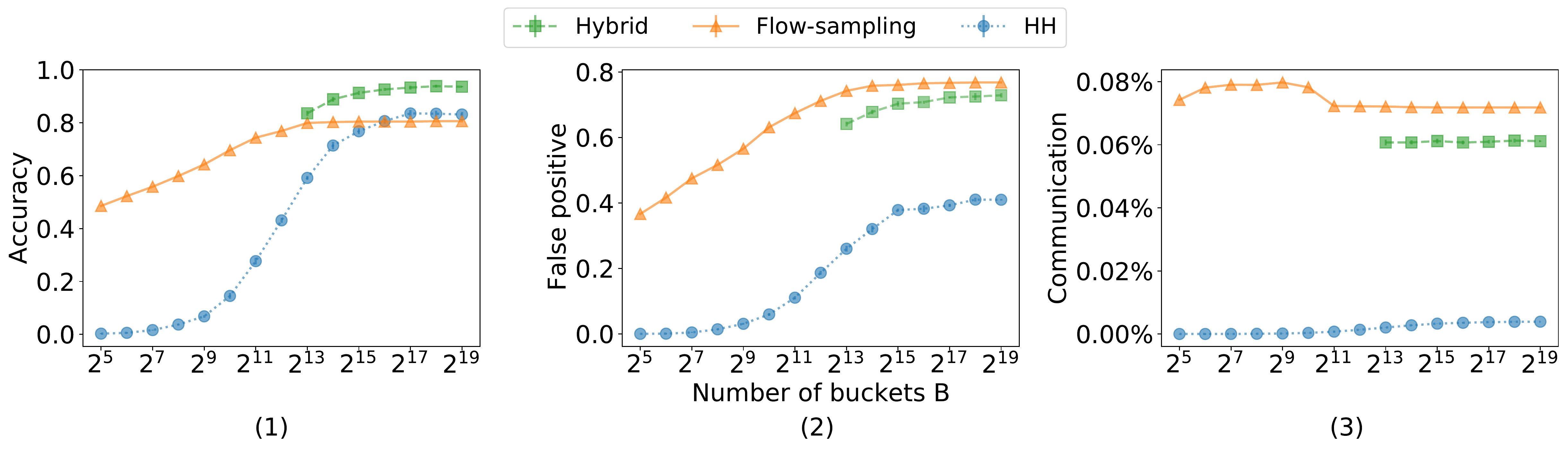}
\vspace{-2.1em}
\caption{
The flow-sampling algorithm achieves great accuracy in small memory ranges, and the hybrid scheme further improves the accuracy when more memory is available.
}
\label{fig:arr_hybrid_hh}
\end{figure*}

\begin{figure*}[t]
\centering
\includegraphics[width=\textwidth]{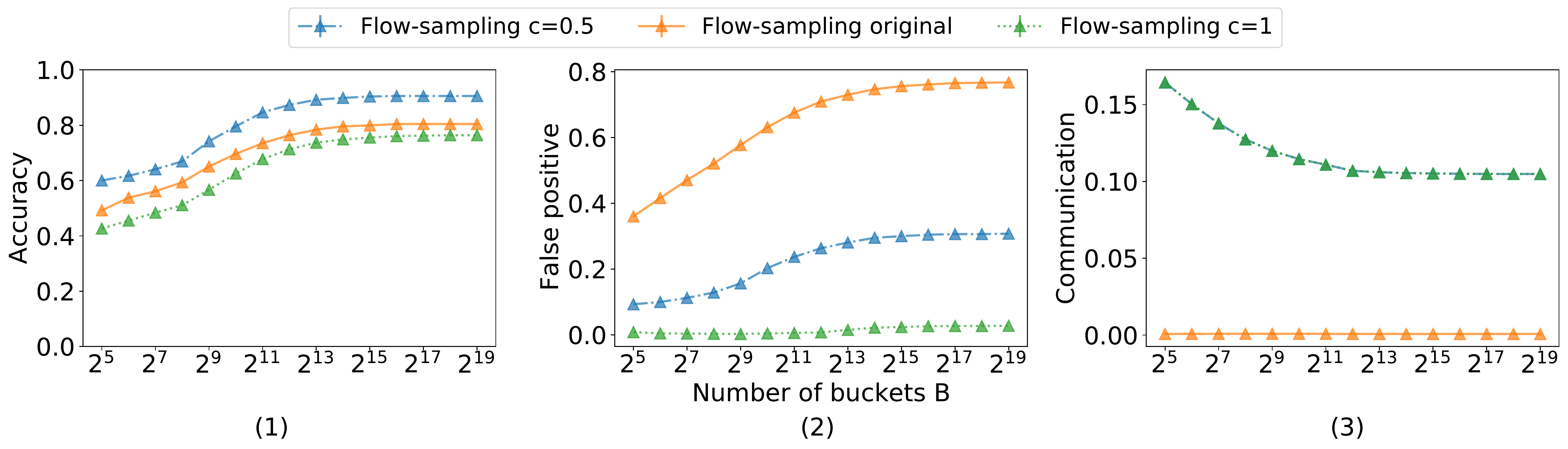}
\vspace{-2.1em}
\caption{
Through sending more reports to the control plane, we can decrease the false-positive rate of the flow-sampling algorithm while further improving its accuracy.
}
\label{fig:arr_eval}
\vspace{-1em}
\end{figure*}

\paragraph{Flow-sampling algorithm achieves great accuracy with small memory.}

If heavy reordering were concentrated in large flows, the HH data structure would perform very well with a small amount of memory.
As seen in \S~\ref{sec:traffic:dist}, real-world traffic does not always behave in that way, rendering the HH data structure ineffective when the memory is small compared to the number of prefixes ($2^{14}$).
This is where the performance of the flow-sampling algorithm significantly dominates that of the HH data structure.
Note that this particular trace contains more than $2^{19}$ flows and more than $2^{14}$ prefixes. 
However, using only $2^5$ buckets, the original version of the flow-sampling algorithm is already capable of reporting half of the out-of-order prefixes. 
To put it into perspective, reordering happens at the flow level, and assigning even one bucket per prefix to detect reordering already requires a nontrivial solution, while the flow-sampling algorithm achieves good accuracy using orders-of-magnitude less memory.

If we are willing to generate reports for more than $10\%$ of the traffic, with an increased communication overhead comes a reduced false-positive rate (Figure~\ref{fig:arr_eval}).
Moreover, with a more carefully chosen parameter $c$ that controls how many prefixes to report (\S~\ref{sec:algo:array:all}), the extra information sent to the control plane helps in further improving the accuracy.

\paragraph{The hybrid scheme improves the accuracy when given more memory.}

To fairly compare the hybrid scheme with the flow-sampling algorithm, we need to determine the optimal memory allocation between the HH data structure and the array.
Lacking a better way to optimize the memory allocation, we turn to experiments with our packet trace.
Given a total of $B$ buckets, we assign $\floor{x \cdot B}$ buckets to the HH data structure, $B - \floor{x B}$ buckets to the array, and conduct a grid search on $x \in I = \{0.1, \dots, 0.9\}$ to find the value of $x$ that maximizes the performance of the hybrid scheme.
We evaluate the hybrid scheme using the optimal $x$ we found for each $B$.

Admittedly, grid $I$ may not be fine-grained enough to reveal the true optimal allocation\edit{;} nonetheless, it conveys the main idea.
When available memory is small, the accuracy gap between the HH data structure and the flow-sampling algorithm is huge, sparing part of the memory for filtering large flows does not improve over the flow-sampling algorithm.
As memory increases, the accuracy gap between the flow-sampling algorithm and the HH data structure decreases, and the hybrid scheme starts to show accuracy gains.

\subsubsection{Performance discrepancies of the flow-sampling algorithm under different workloads}

\begin{figure*}[t]
\centering
\begin{subfigure}[t]{0.32\linewidth}
\centering
\includegraphics[width=\textwidth]{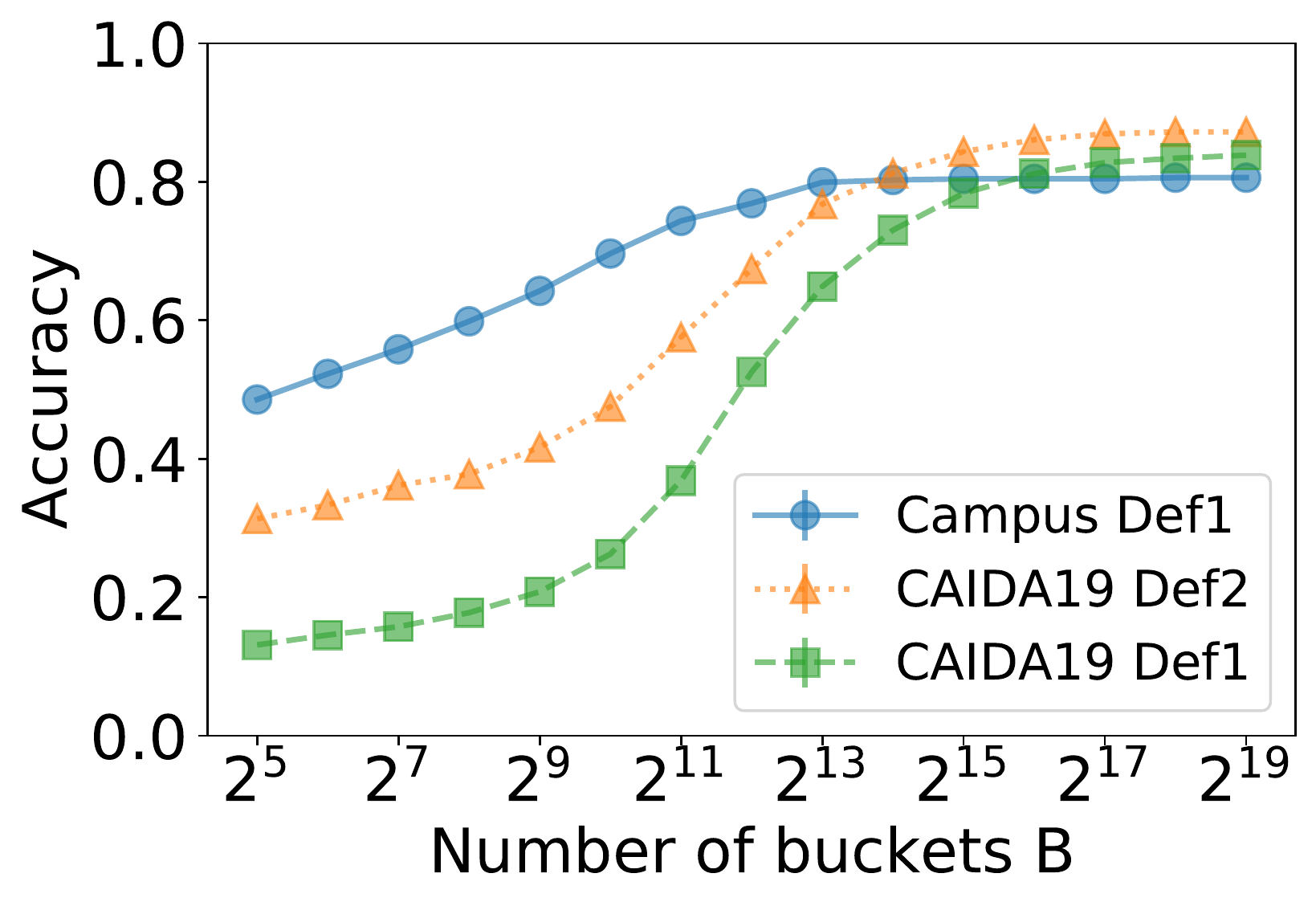}
\caption{
The accuracy of the flow-sampling algorithm may differ under different workloads.
}
\label{fig:array_accuracy_3traces}
\end{subfigure} \hfill
\begin{subfigure}[t]{0.32\linewidth}
\centering
\includegraphics[width=\textwidth]{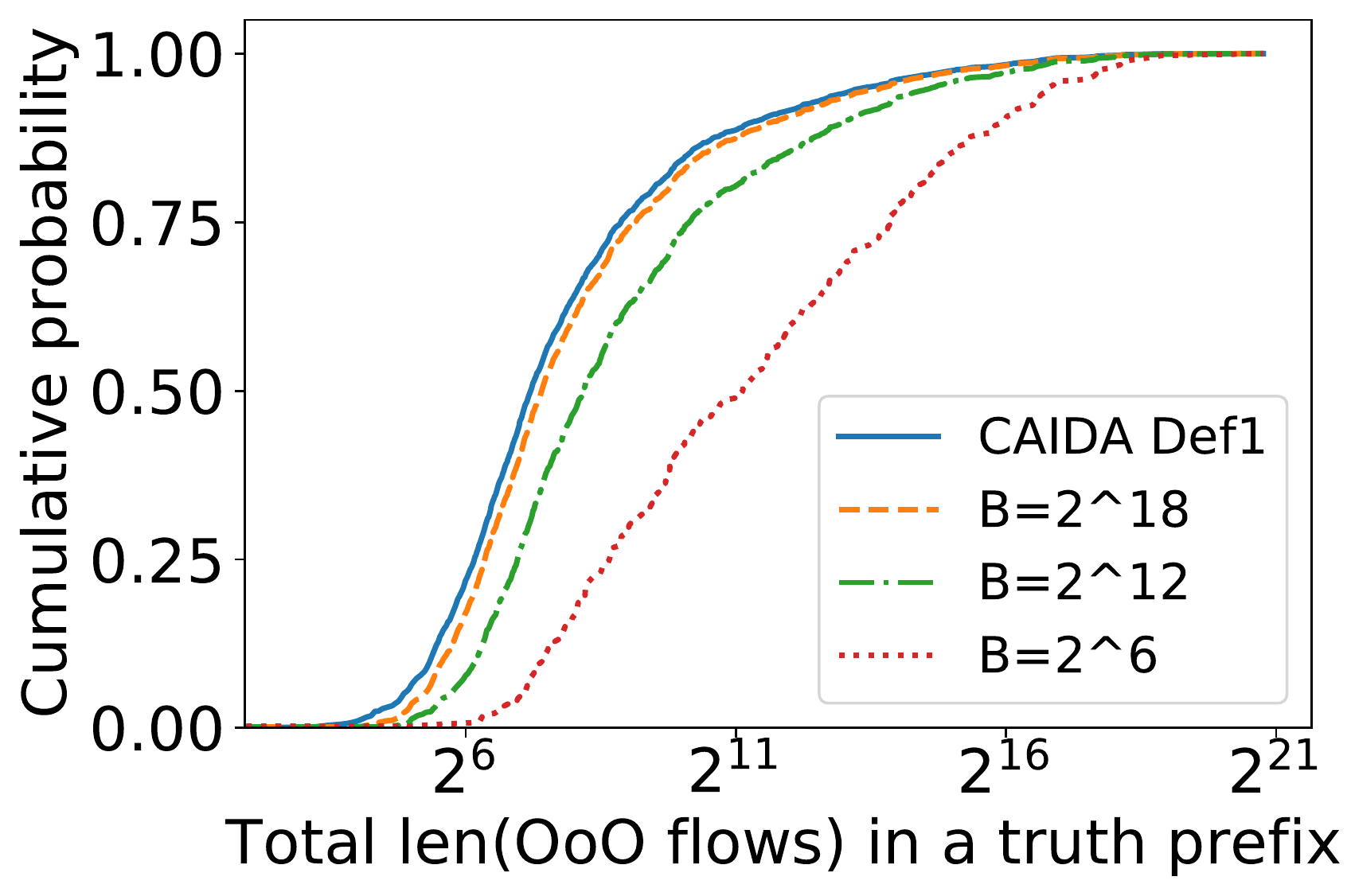}
\caption{
A heavily reordered prefix is easier to capture if the total length of its flows with reordered packets is longer.
}
\label{fig:cdf_len_reordered_flows}
\end{subfigure} \hfill
\begin{subfigure}[t]{0.32\linewidth}
\centering
\includegraphics[width=\textwidth]{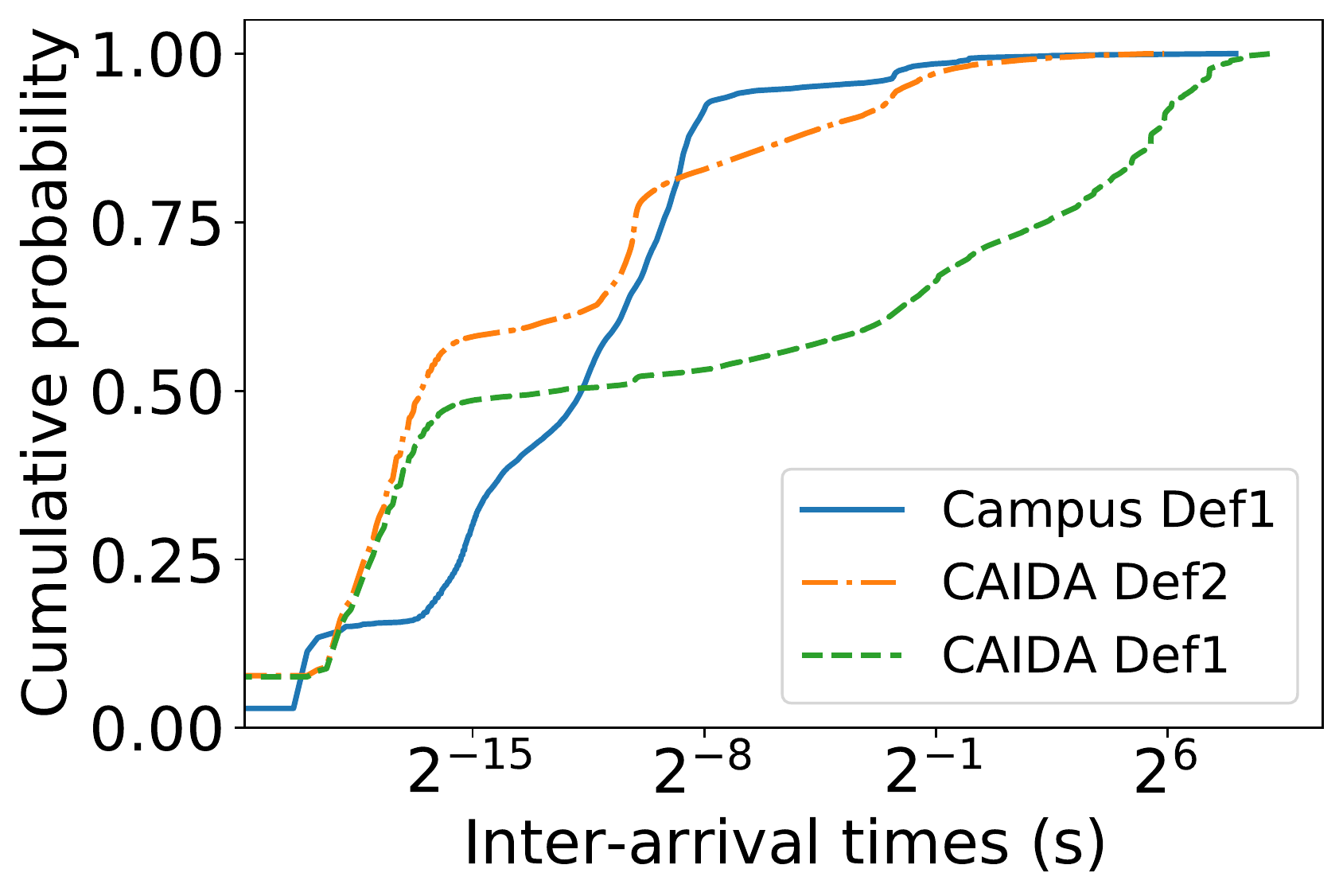}
\caption{
The algorithm is more accurate with small memory when reordered packets arrive shortly after their predecessors.
}
\label{fig:inter_ooo}
\end{subfigure}
\vspace{-0.9em}
\caption{
The accuracy of the flow-sampling algorithm is workload dependent.
}
\vspace{-1em}
\end{figure*} 

In our numerous experiments on different traces, the accuracy of the flow-sampling algorithm always dominates that of the HH data structure, when given much less memory than the number of prefixes in the trace.
However, we cannot always expect to catch $50\%$ of the heavily reordered prefixes using just $B = 2^5$ buckets.
For instance, Figure~\ref{fig:array_accuracy_3traces} shows the accuracy of the flow-sampling algorithm when running on a $5$-minute campus trace using \ref{def:ooo_dec}, and a $10$-minute CAIDA 2019~\cite{caida19} trace using \ref{def:ooo_dec} as well as \ref{def:ooo_inc}.
The results are evidently workload-dependent, but what exactly are the traffic characteristics that dictate such performance discrepancies?
The answer to this question epitomizes the intricacies involved in understanding the multiple distributions present in real-world traffic.

To identify the subset of traffic that directly affects accuracy, we go back to how the flow-sampling algorithm reports a prefix.
If we look at a heavily reordered prefix, its flows enter the data structure from time to time. 
But for the algorithm to report it, the array has to see some flows from this prefix that actually have out-of-order packets. 
The perfectly in-order flows would never contribute to the reporting of its prefix. 
Now, suppose the reordered packets appear uniformly at random during the time its flow is being monitored, then what matters is the total length (number of packets) of the flows that have out-of-order packets in this prefix. 
The higher the length, the easier it is for the flow-sampling algorithm to catch it.
This is in fact an indirect implication of Lemma~\ref{lem}.
It can also be seen in Figure~\ref{fig:cdf_len_reordered_flows}, which shows the CDF of the total length of the flows that have out-of-order packets among all heavily reordered prefixes reported by the flow-sampling algorithm using different memory sizes.
The ground-truth prefixes reported by the smallest memory are the easiest to catch, and the total length of out-of-order flows in such prefixes tends to be larger.
As we increase the memory size, the algorithm reports more ground-truth prefixes with shorter total lengths of out-of-order flows.

However, this is not the whole story.
For the traces in Figure~\ref{fig:array_accuracy_3traces}, the CDFs of the total length of reordered flows in ground-truth prefixes turn out to be similar in shape.
So what else in the traffic distribution is causing the difference in accuracy?
The caveat is that reordered packets may not appear uniformly at random, and their inter-arrival times play a major role as well.
For each dataset in Figure~\ref{fig:array_accuracy_3traces}, we plot the inter-arrival times of their out-of-order packets in the ground truth.
We see that the campus trace, for which the flow-sampling algorithm is the most accurate in the small-memory regime, has $85\%$ of its out-of-order packets in the ground truth arriving within $2^{-8.6}$ seconds of its predecessor in the same flow.
In contrast, in the CAIDA trace, more than $15\%$ of the out-of-order packets corresponding to \ref{def:ooo_dec} do not arrive until $32$ seconds after its predecessor's arrival.
When the memory is small, to sieve through many flows and prefixes, we simply cannot afford wasting much time on one flow, since we may then end up missing many out-of-order events with large inter-arrival times.

\subsection{Hardware feasibility} \label{sec:eval:lucid}

\begin{wraptable}{r}{0.5\textwidth}
\centering
\begin{tabular}{ |c|c|c| } 
\hline
\textbf{Resources} & $B=2^{8}$ & $B=2^{16}$ \\
\hline \hline
TCAM               & $19.05\%$  &  $19.05\%$ \\ 
\hline
SRAM               & $5.00\%$   &  $23.93\%$ \\ 
\hline
Hash units         & $16.67\%$  &  $16.67\%$ \\ 
\hline
Instructions       & $15.18\%$  &  $14.29\%$ \\ 
\hline
\end{tabular}
\vspace{0.5em}
\caption{Data-plane resource usage in Tofino1.}
\label{tab:resources}
\vspace{-2em}
\end{wraptable}

We implement a P4 prototype of the flow-sampling algorithm on a Tofino1 switch using $140$ lines of code in Lucid~\cite{sonchack2021lucid}.
The Lucid-compiled P4 program takes $58.33\%$ of the pipeline stages in Tofino1, while manual inspection of the resulting P4 code shows that $25\%$ of them are overhead from the Lucid compiler.
Even with the overhead introduced by the Lucid compiler, using not even $25\%$ of the resources in the first $58.33\%$ pipeline stages, we are able to report $80.44\%$ of the heavily reordered prefixes in the $5$-minute campus trace using \ref{def:ooo_dec}, and $81.19\%$ and $86.08\%$ in the $10$-CAIDA 2019 trace using \ref{def:ooo_dec} and \ref{def:ooo_inc} respectively.
Out of the $58.33\%$ pipeline stages the algorithm makes use of, the resource usage of the prototype with different number of buckets is summarized in Table~\ref{tab:resources}.
In contrast, merely storing the per-flow states for a $10$-minute CAIDA 2019~\cite{caida19} trace could take more register memory than a Tofino1 switch could offer.

\subsection{Parameter robustness} \label{sec:eval:params}

We started the evaluation using reasonably chosen
parameters. 
Now we verify that all parameters in our algorithms are either easily set, or robust to changes.

To reveal how thresholds $T$ and $C$ individually affect the accuracy of the flow-sampling algorithm, ideally we want to fix one of them to infinity, and vary the other.
In this way, only one of them governs the frequency of evictions.
Applying this logic, when studying the effect of $T$ (Figure~\ref{fig:T}), we fix $C$ to a number larger than the length of the entire trace.
We see that as long as $T$ is small, the algorithm samples enough flows, and has high accuracy.

Evaluating the effects on a varying $C$ turns out to be less straight-forward.
If we make $T$ too large, the algorithm generally suffers from extremely poor performance, which makes it impossible to observe any difference that changing $C$ might bring.
If $T$ is too small, the frequency of eviction would be primarily driven by $T$, and $C$ would not have any impact.
And it is not as simple as setting $T$ larger than all inter-arrival times, since eviction only occurs on hash collisions, inter-arrival time alone only paints part of the picture.
All evidence above points to the fact that $T$ is the more important parameter.
Once we have a good choice of $T$, the accuracy boost from optimizing $C$ is secondary.
Armed with this knowledge, we fix a $T = 2^5$, an ad hoc choice that is by no means perfect. 
Yet it is enough to observe (Figure~\ref{fig:C}) that having a small $C$ is slightly more beneficial.

However, $C$ cannot be too small, as inserting a new flow record into the array requires recirculation in the hardware implementation.
Programmable switches generally support recirculating up to $3\%-10\%$ of packets without penalty. 
Here we set $C$ to be $16$, which allows us to achieve line rate.

Given that each non-small flow is continuously monitored for roughly $C=16$ packets at a time, we report its prefix to the control plane when we encounter any out-of-order packet, that is, $R=1$.

%% file: related.tex
\section{Related work} \label{sec:related}

\textbf{Characterization of out-of-orderness on the Internet.}
Packet reordering is first studied in the seminal work by Paxson~\cite{paxson1997end}.
It has since been well understood that packet reordering can be caused by parallel links, routing changes, and the presence of adversaries~\cite{bennett1999packet}.
In typical network conditions, only a small fraction of packets are out-of-order~\cite{paxson1997end,wang2004study}.
However, when the network reorders packets, TCP endpoints may wrongly infer that the network is congested, harming end-to-end performance by retransmitting packets and reducing the sending rate~\cite{bennett1999packet,laor2002effect,leung2007overview}.
Metrics for characterizing reordering are intensively studied in~\cite{morton2006packet} and~\cite{jayasumana2008improved}, though many of the proposed metrics are more suitable for offline analysis.
In addition to the network causing packet reordering, the stream of packets in the same TCP connection can appear out of order because congestion along the path leads to packet losses and subsequent retransmissions.  Our techniques for identifying IP prefixes with heavy reordering of TCP packets are useful for pinpointing network paths suffering from both kinds of reordering---whether caused by the network devices themselves or induced by the TCP senders in response to network congestion.

\textbf{Data-plane efficient data structures for volume-based metrics.}
For heavy-hitter queries, HashPipe~\cite{sivaraman2017heavy} adapts SpaceSaving~\cite{metwally2005efficient} to work with the data-plane constraints, using a multi-stage hash-indexes array.
PRECISION~\cite{basat2020designing} further incorporates the idea of Randomized Admission Policy~\cite{basat2019randomized} to better deal with the massive number of small flows generally found in network traffic.
We extend PRECISION to keep reordering statistics for large flows.
However, such an extension cannot be used to detect flows with a large number of out-of-order packets with a reasonable amount of memory.

\textbf{Data-plane efficient data structures for performance metrics.}
Liu et al.~\cite{liu2020memory} proposes memory-efficient algorithms for identifying flows with high latency, or lost, reordered, and retransmitted packets.
Several solutions for measuring round-trip delay in the data plane~\cite{chen2020measuring,zheng2022unbiased,sengupta2022continuous} have a similar flavor to identifying out-of-order heavy prefixes, as in both cases keeping at least some state is necessary, with the difference that for reordering we generally need to match more than a pair of packets.

\textbf{Detecting heavy reordering in the data plane.}
Several existing systems can detect TCP packet reordering in the data plane.
Marple is a general-purpose network telemetry platform with a database-like query language~\cite{narayana2017language}.
While Marple can analyze out-of-order packets, the compiler generates a data-plane implementation that requires per-flow state.
Unfortunately, such methods consume more memory than the programmable switch can offer in practice.
The algorithm proposed by Liu et al.~~\cite{liu2020memory}
for detecting flows with a large number of out-of-order packets
remains the work most related to ours.
We note that our lower bound on memory consumption in \S~\ref{sec:prob:flow:lb} is stronger than a similar lower bound (Lemma 10) in~\cite{liu2020memory}, as we also allow randomness and approximation.
Liu et al.~\cite{liu2020memory} considers out-of-order events specified by \ref{def:ooo_max}, and works around the lower bound by assuming out-of-order packets always arrive within some fixed period of time.
In contrast, we circumvent the lower bound using the more natural observation that out-of-orderness is correlated among flows within a prefix, and identify heavily reordered prefixes instead of flows.

%% file: conclusion.tex
\section{Conclusion} \label{sec:conclusion}

In this paper, we introduce two algorithms for identifying out-of-order prefixes in the data plane.
In particular, our flow-sampling algorithm achieves good accuracy empirically, even with memory that is orders-of-magnitude smaller than the number of prefixes, let alone the number of flows.
When given memory comparable to the number of prefixes, our  hybrid scheme using both a heavy-hitter data structure and flow sampling slightly improves the accuracy.

Notice that measuring reordering is fundamentally memory-intensive, yet we leverage the correlation of out-of-orderness among flows in the same prefix so that compact data structures can be effective.
In fact, there is nothing special about out-of-orderness.
Other properties of a network path could very well lead to similar correlation.
For other performance metrics that suffer from memory lower bounds, it would be intriguing to see whether such correlation helps in squeezing good performance out of limited memory.
We leave this as future work.

%% file: appendix.tex
\newpage
\section{Track heavy flows over long periods} \label{sec:precision}

\begin{figure*}[t]
\centering
\includegraphics[width=1\textwidth]{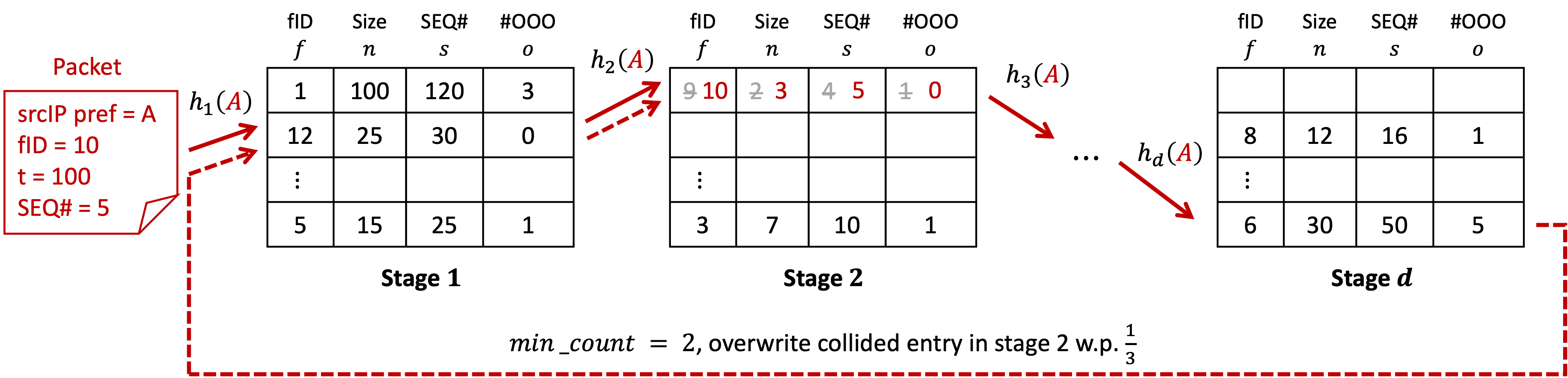}
\caption{
A modification of PRECISION for tracking out-of-order packets.
}
\label{fig:precision}
\end{figure*}

\begin{figure*}[t]
\centering
\includegraphics[width=\textwidth]{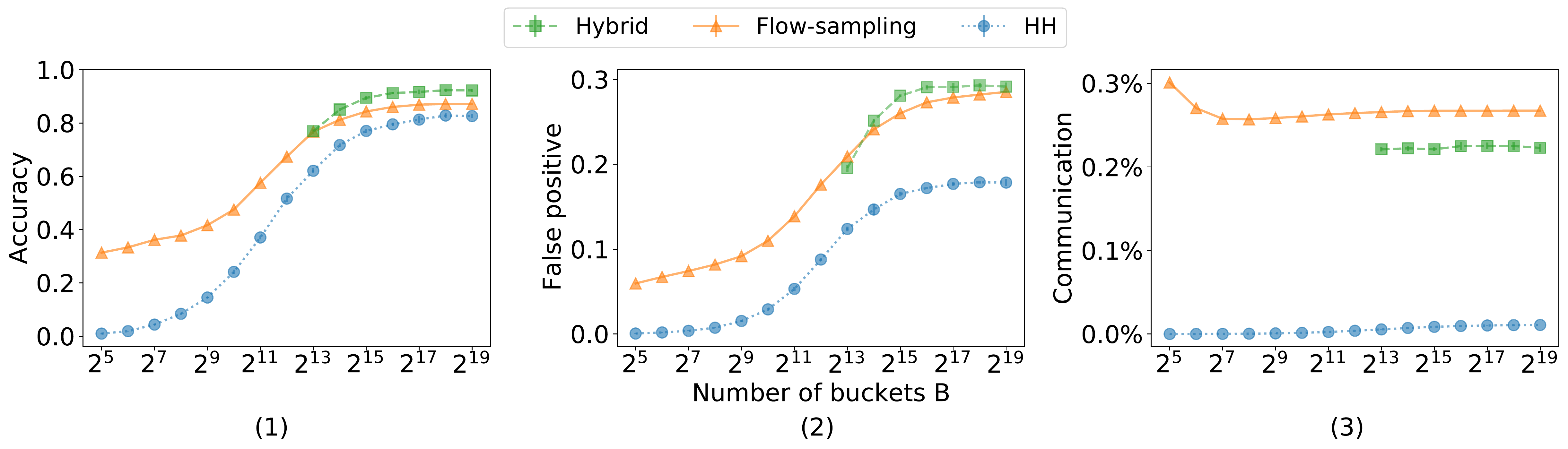}
\vspace{-2em}
\caption{Performance of the flow-sampling algorithm, the hybrid scheme and the HH data structure on a $10$-minute CAIDA 2019 trace for ~\ref{def:ooo_inc}.
}
\label{fig:arr_hybrid_hh_caida}
\end{figure*}

\begin{figure*}[t]
\centering
\begin{subfigure}[t]{0.32\linewidth}
\centering
\includegraphics[width=\textwidth]{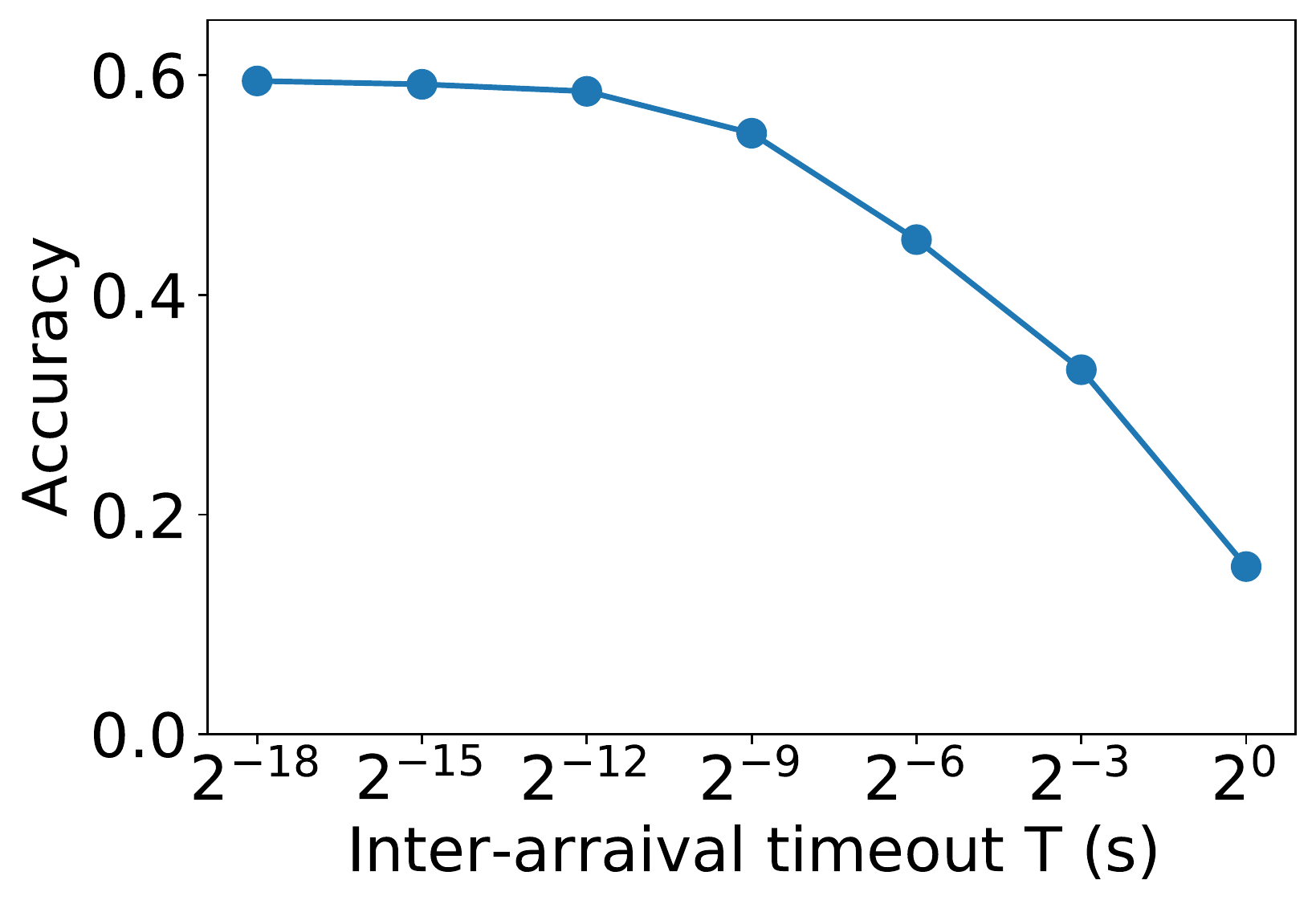}
\caption{The accuracy of the flow-sampling algorithm with varying $T$, and fixed $B = 2^8$, $R = 1$ and $C = 10^8$.}
\label{fig:T}
\end{subfigure} \hfill
\begin{subfigure}[t]{0.32\linewidth}
\centering
\includegraphics[width=\textwidth]{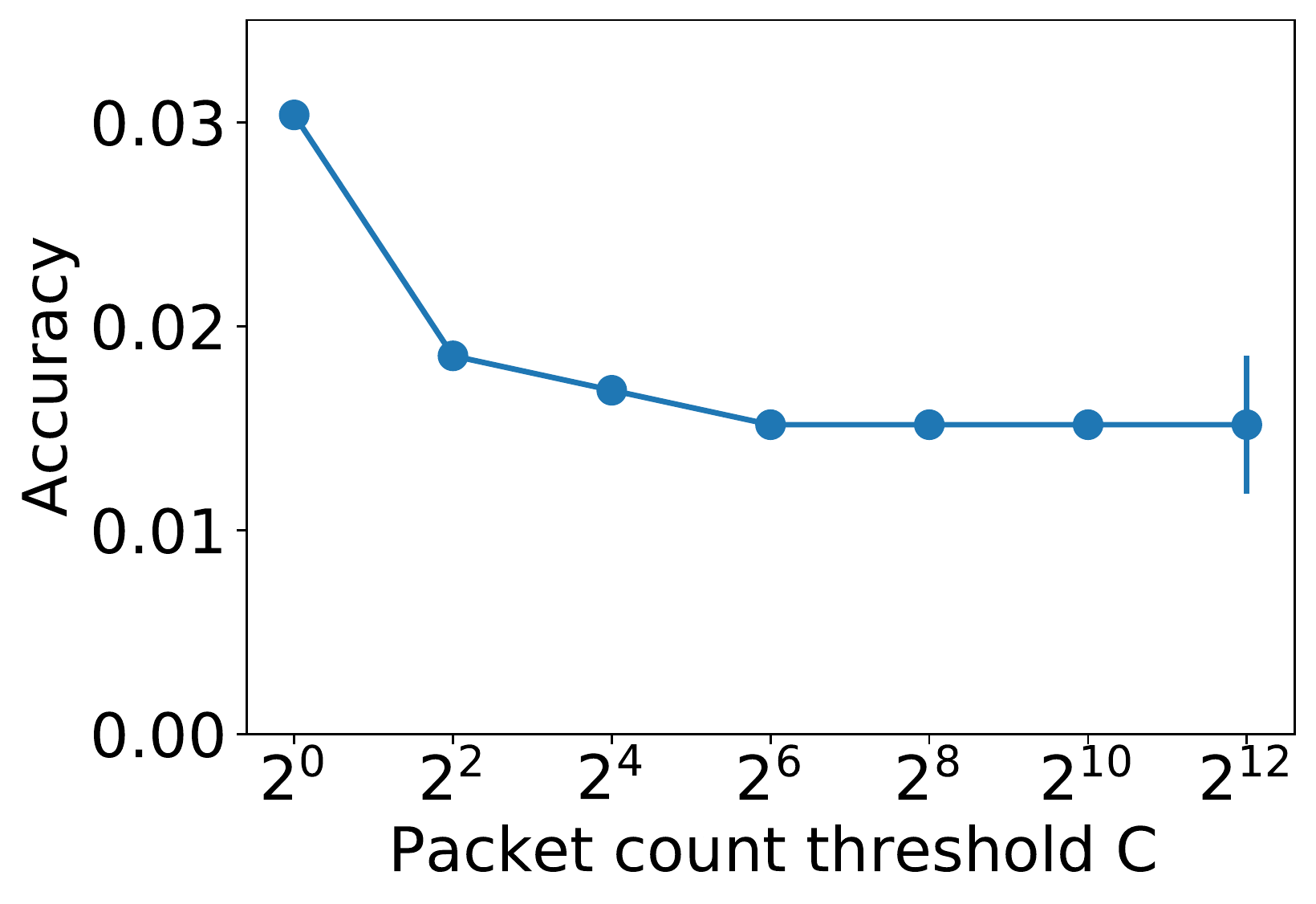}
\caption{The accuracy of the flow-sampling algorithm with varying $C$, with fixed $B = 2^8$, $R = 1$ and $T = 2^5$.}
\label{fig:C}
\end{subfigure} \hfill
\begin{subfigure}[t]{0.32\linewidth}
\centering
\includegraphics[width=\textwidth]{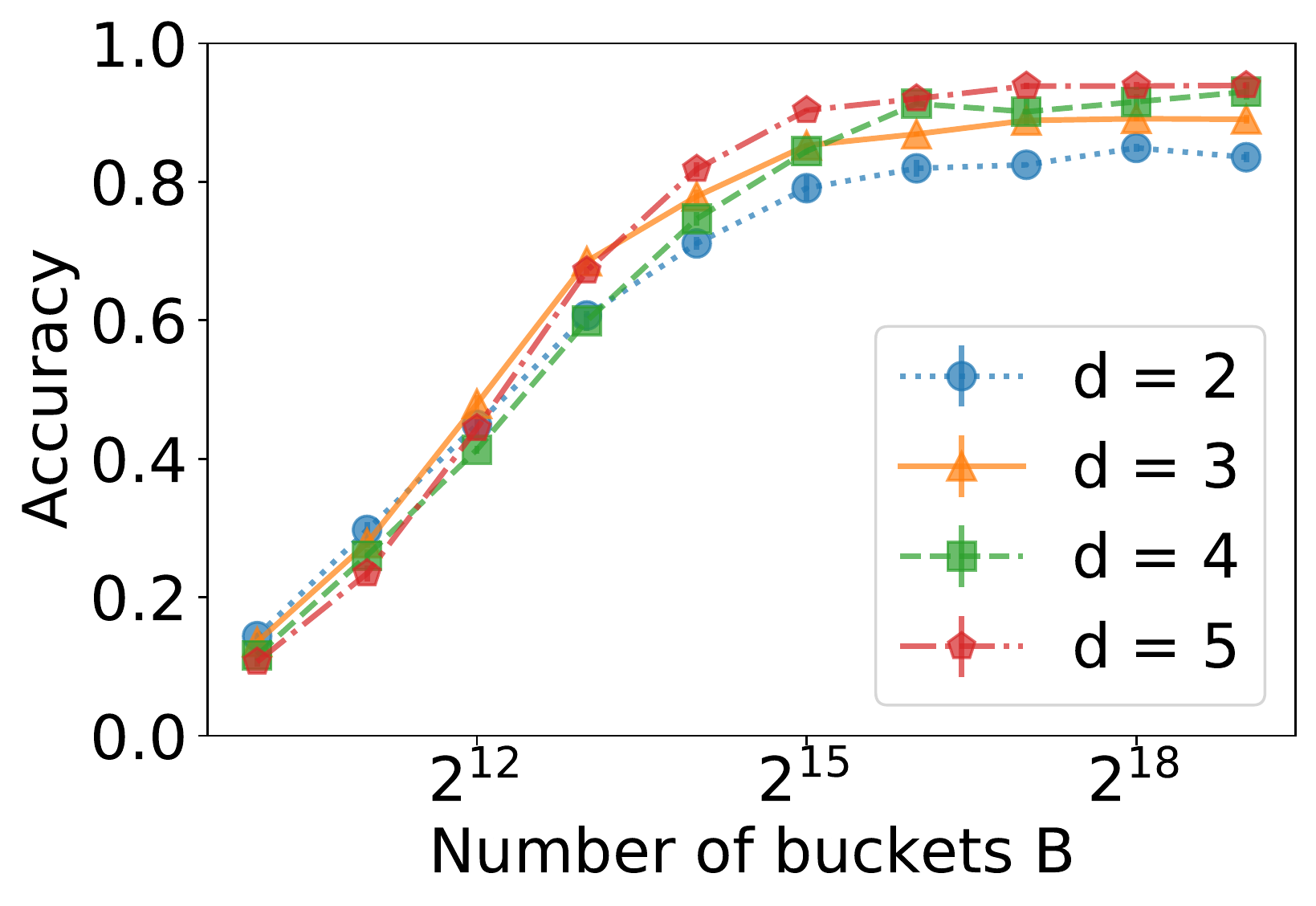}
\caption{The accuracy of the flow-sampling algorithm with varying $d$, with fixed $R = 0.01$.}
\label{fig:d}
\end{subfigure}
\caption{The effect of changing parameters on the accuracy of the flow-sampling algorithm and PRECISION.}
\end{figure*}

To capture out-of-orderness in heavy flows, we want a data structure that is capable of simultaneously tracking heaviness and reordering.
The SpaceSaving~\cite{metwally2005efficient} data structure fits naturally for the task, as we can maintain extra state for each flow record, while the data structure gradually identifies the flows with heavy volume by keeping estimates of their traffic counts.
However, when overwriting a flow record to admit a new flow, SpaceSaving needs to go over all entries to locate the flow with the smallest traffic count, which makes it infeasible for the data plane due to the constraint on the number of memory accesses per packet.

Thus, we opt for PRECISION~\cite{basat2020designing}, the data-plane adaptation of SpaceSaving, which checks only a small number of $d$ entries when overwriting a flow record.
We emphasize that the specifics about how PRECISION works are not, in fact, important in this context.
It is enough to bear in mind that with a suitable data-plane friendly heavy-hitter algorithm, tracking reordering is exactly the same as in the strawman solution (\S~\ref{sec:prob:flow:strawman}), but applied only to heavy flows.
Figure~\ref{fig:precision} shows the modified PRECISION for tracking out-of-order packets using $d$ stages.

We again assign flows from the same prefix to the same set of buckets, by hashing prefixes instead of flow IDs.
In a PRECISION data structure with $d$ stages, at the end of the stream, at most $d$ heaviest flows from each prefix $g$ would remain in memory.
Doing so effectively frees up buckets that used to be taken by a few prefixes with many heavy flows, and allows more prefixes to have their heaviest flows measured.

\section{Supplementary evaluation} \label{sec:supp_eval}

\paragraph{More on performance evaluation}

To show that the overall trend of the performance curves in Figure~\ref{fig:arr_hybrid_hh} is representative, Figure~\ref{fig:arr_hybrid_hh_caida} presents the performance of proposed algorithms on a $10$-minute CAIDA 2019~\cite{caida19} trace using ~\ref{def:ooo_inc}.
The trace contains $61,791,947$ server-to-client packets that come from $2,717,709$ flows and $54,148$ $24$-bit source IP prefixes.

The flow-sampling algorithm remains the most effective when given a small amount of memory, and the hybrid scheme achieves the best accuracy when more memory is available.
Interestingly, in this case, the accuracy of the flow-sampling algorithm dominates that of the HH data structure even when the memory is comparable to the number of prefixes.

\paragraph{More on parameter robustness}

It is observed in~\cite{basat2020designing} that a small constant $d > 1$ only incurs minimal accuracy loss in finding heavy flows.
Increasing $d$ leads to diminishing gains in performance, and adds the number of pipeline stages when implemented on the hardware.
Therefore, $d=2$ is preferable for striking a balance between  accuracy and hardware resources.

Building on~\cite{basat2020designing}, we evaluate PRECISION for $d=2, 3, 4, 5$, for reporting out-of-order heavy prefixes.
The results in Figure~\ref{fig:d} show that when the total memory is small, using fewer stages provides a slight benefit.
The opposite holds when there is ample memory.
However, as the performance gap using different $d$ is insignificant, we also suggest using $d=2$ for hardware implementations.